\documentclass[11pt]{article}

\usepackage[T1]{fontenc}
\usepackage[utf8]{inputenc}
\usepackage{lmodern}
\usepackage[left=1.08in,right=1.08in,top=0.82in,bottom=0.90in]{geometry}
\usepackage{amsmath,amssymb,amsthm,mathtools}
\usepackage{booktabs,array}
\usepackage{graphicx}
\usepackage{flafter}
\usepackage[section]{placeins}
\usepackage{xcolor}
\usepackage{microtype}
\usepackage[numbers,square,sort&compress]{natbib}
\usepackage{hyperref}
\definecolor{linkred}{RGB}{155,0,0}
\hypersetup{
  colorlinks=true,
  linkcolor=linkred,
  citecolor=linkred,
  urlcolor=linkred,
  pdftitle={Certified Alpha Capacity: Statistical Arbitrage When Learning Takes Time},
  pdfauthor={Nicolò Bonacorsi},
  pdfsubject={Statistical certification of finite-lived trading signals},
  pdfkeywords={Certified Alpha Capacity, alpha decay, sequential testing, information theory, quantitative finance}
}
\ifdefined\pdfinfoomitdate\pdfinfoomitdate=1\fi
\ifdefined\pdftrailerid\pdftrailerid{}\fi
\ifdefined\pdfsuppressptexinfo\pdfsuppressptexinfo=15\fi
\usepackage{url}
\numberwithin{equation}{section}
\theoremstyle{plain}
\newtheorem{theorem}{Theorem}[section]
\newtheorem{proposition}[theorem]{Proposition}

\newtheorem{corollary}[theorem]{Corollary}
\theoremstyle{definition}

\theoremstyle{remark}
\newtheorem{remark}[theorem]{Remark}

\DeclareMathOperator{\KL}{D_{\mathrm{KL}}}
\newcommand{\Pzero}{\mathbb P_0}
\newcommand{\Pone}{\mathbb P_1}
\newcommand{\Ezero}{\mathbb E_0}
\newcommand{\Eone}{\mathbb E_1}
\newcommand{\Fobs}{\mathbb F}
\newcommand{\Ginfo}{\mathbb G}

\makeatletter
\renewcommand{\@seccntformat}[1]{\csname the#1\endcsname.\quad}
\makeatother

\begin{document}

\begin{center}
\vspace*{0.20em}
{\fontsize{24}{28}\selectfont\bfseries Certified Alpha Capacity:\par}
\vspace{0.35em}
{\fontsize{15}{18}\selectfont\scshape
Statistical Arbitrage\\[-0.08em]
When Learning Takes Time\par}

\vspace{1.15em}
{\large\scshape Nicolò Bonacorsi\par}
\vspace{0.20em}
{\normalsize Department of Applied Physics and Applied Mathematics\par}
{\normalsize Columbia University\par}
{\small New York, NY 10027, USA\par}
\vspace{0.30em}
{\small \texttt{nb3328@columbia.edu}\par}
{\small ORCID: \href{https://orcid.org/0009-0005-0479-3102}{0009-0005-0479-3102}\par}
\end{center}

\vspace{0.65em}
\begin{abstract}
Statistical validation takes time, and alpha can decay before the evidence justifies deployment. We quantify the surviving opportunity through Certified Alpha Capacity: the maximum expected value remaining under prescribed false-deployment and power constraints. In a canonical Gaussian experiment, we derive an exact certification frontier and the sharp law $\mathcal K_{\alpha,\beta}/\log(1/\varepsilon)$ for residual information capacity just above it. Near the frontier, a single interim decision can preserve orders of magnitude more capacity than fixed-time testing. Mapping information into cash reveals why equally certifiable signals can retain different economic values. In the specified competitive equilibrium, lifetime information approaches the certification frontier exponentially as research costs vanish.
\end{abstract}

\vspace{0.30em}

\vspace{0.40em}

\vspace{0.45em}

\section{Introduction}\label{introduction}

A promising trading signal creates a timing problem before it creates a portfolio position. Evidence takes time to accumulate, and the opportunity can weaken during that time. A large current signal may disappear before it supports a reliable deployment decision; a smaller, persistent signal may permit enough observation to justify action. Statistical arbitrage therefore requires assessing both the evidence needed to act and the opportunity remaining when that evidence arrives.

Research on post-publication attenuation, slow exploitation, and alpha dynamics documents the importance of persistence for investment opportunities \cite{mcleanpontiff2016,jacobsmuller2020,kaplanski2023,penasse2022}. Selecting among many candidate signals also increases the burden of establishing a reliable discovery \cite{harveyliu2016,baileyloprado2014,baileyborwein2017}. Together, these mechanisms make the value available after validation depend on the relative rates of learning and economic decay.

We study this dependence through \emph{Certified Alpha Capacity} (CAC). A certification contract fixes a maximum false-deployment probability under an economic boundary model and a required deployment probability under the alternative. Kullback--Leibler information measures separation from that boundary and provides a clock for accumulated evidence and remaining information lifetime. A stopping rule chooses when to certify, subject to the contract, and receives the value left at that decision. The framework joins a reliability requirement to the reward from timely action within a specified observation experiment.

Sequential testing and group-sequential design supply the statistical foundation. Likelihood processes control errors as observations arrive \cite{wald1947,siegmund1985}; error spending and optimal designs allocate the reliability budget across monitoring dates \cite{obrienfleming1979,landemets1983,ealesjennison1992,barberjennison2002}. Opportunistic detection and time-sensitive testing make terminal constraints central to the decision \cite{zhangmoustakidespoor2016,clericoetal2026,hamjasinzhao2026}. Our central result quantifies the value reliable stopping can preserve when the lifetime is only just sufficient for certification: a sharp inverse-logarithmic law with matching bounds over admissible stopping rules (Theorem~\ref{thm:near-critical-capacity}). A small surplus of lifetime information has a quantitative value given by this explicit stopping-value law.

The implications connect statistical feasibility to arbitrage economics. At false-deployment level $0.05$ and power $0.90$, the sharp constant is $\mathcal K_{0.05,0.10}\simeq11.15$. In the zero-hurdle exponential benchmark, a Bonferroni search over 100 candidates with annualized instantaneous Sharpe 2 requires a half-life of about 7.24 years (Section~\ref{alpha-survival-frontiers}). An ordered family of unknown decay profiles retains an exact least-profile frontier (Theorem~\ref{thm:ordered-decay-frontier}), while competitive entry drives excess lifetime information exponentially toward that frontier as research costs vanish (Corollary~\ref{cor:near-critical-equilibrium}). A descriptive funding benchmark illustrates the link between estimated lifetime information and persistence (Section~\ref{retrospective-funding-study}).

Economic interpretation requires a payoff map. Information in a Gaussian drift experiment is quadratic in the drift gap, whereas expected cashflow from a bounded position is linear in that gap. Equal lifetime information can thus support different cash values after learning. This connects CAC to learning with deadlines and decision-valued experimental design \cite{frazier2007deadline,dayanik2013deadline,campbell2026deadline,anand2010,zhong2026}. Information acquisition and arbitrage can also change other participants' opportunities \cite{grossmanstiglitz1980,basakcroitoru2006,zigrand2004,zigrand2006}, and statistical limits of arbitrage connect econometric discrimination to market opportunities \cite{danagelxiu2024}. CAC isolates the timing constraint within a declared observation law, payoff map, and lifetime mechanism.

\section{Model}\label{model-and-information-time}

\subsection{Probabilistic setup and information time}\label{probabilistic-setup}

On a common observation space, let $\Fobs=(\mathcal F_t)$ be the usual right-continuous augmentation of $\sigma\{(X_u,Y_u):u\le t\}$, and let $\Pzero,\Pone$ be locally equivalent laws. The signal $X$ and loading $\theta$ are predictable, and
\begin{equation}\label{eq:observed-channel-alt}
dY_t=\theta_tX_t\,dt+\sigma\,dB_t^{(1)}\quad(\Pone),\qquad \sigma>0,
\end{equation}
\begin{equation}\label{eq:observed-channel-null}
dY_t=\sigma\,dB_t^{(0)}\quad(\Pzero),
\end{equation}
where each $B^{(i)}$ is Brownian in the stated law and filtration. We impose local square integrability and the exponential-martingale conditions for Girsanov (Novikov suffices on bounded horizons). The joint law of $(X,Y)$ is assumed to have the likelihood below: the initial observation has the same law and the predictor supplies no additional likelihood factor. Predictability alone does not imply this joint-law assumption. A common completion is used on each horizon.

Define the continuous information clock and its stopped inverse by
\begin{equation}\label{eq:information-clock}
A_t=\sigma^{-2}\int_0^t\theta_u^2X_u^2\,du,\qquad H=A_T,
\end{equation}
\begin{equation}\label{eq:inverse-information-clock}
t_{\rm inv}(s)=\inf\{t:A_t>s\}\wedge T,\qquad \mathcal G_s=\mathcal F_{t_{\rm inv}(s)}.
\end{equation}
Here $T$ is economic death, deterministic or an $\Fobs$-stopping time; $A$ denotes the value of $H$ when deterministic. The information filtration $\Ginfo=(\mathcal G_s)$ is constant after $H$, and flat calendar-clock intervals are collapsed. Girsanov gives the log likelihood
\begin{align}
\Lambda_t&=\int_0^t\frac{\theta_uX_u}{\sigma}\,dB_u^{(0)}-A_t/2 &&(\Pzero),\label{eq:llr-null}\\
\Lambda_t&=\int_0^t\frac{\theta_uX_u}{\sigma}\,dB_u^{(1)}+A_t/2 &&(\Pone).\label{eq:llr-alt}
\end{align}
Thus $Z_s=\Lambda_{t_{\rm inv}(s)}$ satisfies, up to $H$,
\begin{equation}\label{eq:information-time-likelihood}
Z_s=\widetilde B_s^{(1)}+s/2\quad(\Pone),\qquad Z_s=\widetilde B_s^{(0)}-s/2\quad(\Pzero),
\end{equation}
\begin{equation}\label{eq:information-time-rn}
\left.\frac{d\Pone}{d\Pzero}\right|_{\mathcal G_s}=e^{Z_s}.
\end{equation}
At a random horizon the precise identity is stopped: $Z_s=\widetilde B^{(i)}_{s\wedge H}\pm(s\wedge H)/2$. Any Brownian extension after $H$ is auxiliary on an enlarged space. For a bounded $\Ginfo$-stopping time $S$, write $\mathbb P_i^S$ for the restriction to its stopped sigma-field $\mathcal G_S$, and $\mathcal T_A$ for the stopping times $0\le S\le A$.

We use $\KL(P\Vert Q)=\mathbb E_P\log(dP/dQ)$ and $\operatorname{kl}(p,q)=p\log(p/q)+(1-p)\log((1-p)/(1-q))$, extended to boundary arguments by lower semicontinuity. All logarithms are natural; $\Phi$ is the standard normal cdf, $\bar\Phi=1-\Phi$, $z_u=\Phi^{-1}(u)$, and $x_+=\max(x,0)$. Matrix inequalities use the positive-definite order. Principal model-specific notation is collected in Appendix~\ref{app:notation}.

\subsection{Economic value on the information clock}\label{economic-value-information-clock}

At a fixed calendar time, suppress the time subscript and write $X$ and $\theta$
for the current signal value and drift loading. For the quadratic instantaneous trading objective
\[
q\theta X-\frac{\gamma}{2}q^2,
\qquad \gamma>0,
\]
where \(q\in\mathbb R\) is the position and \(\gamma\) is the quadratic trading-cost coefficient, the full-information optimal position is \(q^*=\theta X/\gamma\). Its value rate is
\[
\frac{\theta^2X^2}{2\gamma}\,dt
=\frac{\sigma^2}{2\gamma}\,dA_t.
\]
Thus statistical information and perfect-information economic value are carried by the same clock. When $H=A$ is deterministic, define the lifetime path information and integrated perfect-information value by
\[
\mathcal I_{\rm life}:=\KL(\Pone^A\Vert\Pzero^A),
\qquad
V^{PI}:=\int_0^T\frac{\theta_t^2X_t^2}{2\gamma}\,dt,
\]
where the superscript $PI$ stands for perfect information.

\begin{theorem}[Statistical information and full-information value]
In the deterministic information-horizon benchmark $H=A$,
\[
\mathcal I_{\rm life}=\frac A2,
\qquad
V^{PI}=\frac{\sigma^2}{2\gamma}A
=\frac{\sigma^2}{\gamma}\mathcal I_{\rm life}.
\]
More generally, for any bounded $\Ginfo$-stopping time $S$ satisfying $S\le H$ almost surely,
\[
\KL(\Pone^S\Vert\Pzero^S)=\Eone[Z_S]
=\frac12\Eone[S].
\]
\end{theorem}

\begin{proof}
For every bounded stopping time \(S\), the stopped likelihood-ratio identity gives
\[
\left.\frac{d\Pone}{d\Pzero}\right|_{\mathcal G_S}=e^{Z_S}.
\]
Therefore
\[
\KL(\Pone^S\Vert\Pzero^S)=\Eone[Z_S].
\]
Under \(\Pone\), \(Z_s=\widetilde B_s^{(1)}+s/2\). Boundedness of \(S\) makes optional sampling applicable to \(\widetilde B^{(1)}\), so \(\Eone[\widetilde B_S^{(1)}]=0\) and hence \(\Eone[Z_S]=\Eone[S]/2\). For deterministic \(S=A\), this gives \(\mathcal I_{\rm life}=A/2\). The economic identity follows from the full-information quadratic value calculation immediately above.
\end{proof}

Under this reduction, signal amplitude, feature units, calendar half-life and lookback enter the certification problem through the total information generated before the economic endpoint.

\section{Reliable certification before decay}\label{reliable-deployment-and-certified-alpha-capacity}

Fix error parameters $\alpha,\beta\in(0,1)$, where $\alpha$ is the maximum
false-deployment probability under $\Pzero$ and $1-\beta$ is the required
deployment probability under $\Pone$.  Unless stated otherwise we consider the
nondegenerate regime $1-\beta\ge\alpha$.

With only a type-I constraint, a procedure can deploy immediately with
probability \(\alpha\) without learning anything about the signal. A power target strictly above $\alpha$ rules out this lottery and makes certification depend on evidence
accumulated under the alternative.

Throughout the canonical experiment, the observation filtration is augmented at time zero by a variable $U_{\rm rnd}\sim\mathrm{Unif}(0,1)$, independent of the observed path under both hypotheses and with the same law under both. This permits binary randomized decisions without adding information about the hypothesis.

For deterministic information lifetime \(A\), a certification/deployment rule is a pair \((S,\delta)\), where \(S\in\mathcal T_A\) and \(\delta\in\{0,1\}\) is \(\mathcal G_S\)-measurable. The event \(\{\delta=1\}\) means that the signal is certified and deployed. On \(\{\delta=0\}\) we may, without loss, set \(S=A\). Define
\[
c_{\alpha,\beta}(A)
:=\sup \Eone[(A-S)\delta]
\]
over all such rules satisfying
\[
\Pzero(\delta=1)\le\alpha,
\qquad
\Pone(\delta=1)\ge1-\beta.
\] Below the feasibility threshold, we extend the capacity by zero as a reporting convention; the no-deployment rule does not satisfy the power constraint. The corresponding expected post-certification value, measured in economic units, is \[
\mathsf{CAC}_{\alpha,\beta}(A)
=\frac{\sigma^2}{2\gamma}c_{\alpha,\beta}(A).
\]

\subsection{Exact Gaussian certification threshold}\label{exact-reliable-information-phase-transition}

\begin{theorem}[Exact Gaussian certification threshold]\label{thm:exact-gaussian-threshold}
In the canonical Brownian likelihood experiment, the largest possible probability of deployment by information time $A$ among all level-$\alpha$ rules is
\[
\pi_\alpha(A)=\Phi\!\left(\sqrt A-z_{1-\alpha}\right).
\]
Consequently a false-deployment target $\alpha$ and power target $1-\beta$ are statistically feasible by death if and only if
\[
A\ge A_{\rm crit}(\alpha,\beta)
=\bigl(z_{1-\alpha}+z_{1-\beta}\bigr)^2.
\]
Equivalently, in KL units,
\[
\mathcal I_{\rm life}
\ge
\mathcal I_{\rm crit}(\alpha,\beta)
:=\frac12\bigl(z_{1-\alpha}+z_{1-\beta}\bigr)^2.
\]
Positive post-certification value requires $A>A_{\rm crit}$: at the critical point the full lifetime is needed merely to attain the requested statistical discrimination.
\end{theorem}

\begin{proof}
If $A=0$, the two observation laws coincide. No level-$\alpha$ decision has power exceeding $\alpha$, and the decision $\mathbf1_{\{U_{\rm rnd}\le\alpha\}}$ attains it. Necessarily $S=0$, so residual value is zero. Thus the displayed formula holds at $A=0$. Assume $A>0$ for the Gaussian threshold argument below.

Every decision made by a rule with \(S\le A\) is \(\mathcal G_A\)-measurable, so the Neyman--Pearson lemma applies to the full experiment observed up to \(A\). The full path likelihood ratio at time \(A\) is \(e^{Z_A}\), and \(Z_A\) is sufficient for testing the two simple path laws. Under \(\Pzero\) and \(\Pone\), \[
Z_A\sim N(-A/2,A),
\qquad
Z_A\sim N(A/2,A),
\] respectively. Neyman--Pearson gives the level-\(\alpha\)
threshold
\[
c_A:=-\frac A2+\sqrt A\,z_{1-\alpha},
\qquad Z_A\ge c_A,
\]
whose \(\Pone\)-probability equals \(\Phi(\sqrt A-z_{1-\alpha})\).
Solving \(\pi_\alpha(A)\ge1-\beta\) gives the displayed threshold. If
\(A>A_{\rm crit}\), choose a deterministic test time
\(t\in(A_{\rm crit},A)\); the level-\(\alpha\) Neyman--Pearson test at \(t\) has
power greater than \(1-\beta\) and leaves positive time \(A-t\) whenever
it deploys. At the critical point with $A_{\rm crit}>0$, feasibility forces the rule to attain the extremal
Neyman--Pearson receiver-operating-characteristic (ROC) point. Because the terminal likelihood ratio has a continuous
distribution, the level-\(\alpha\) most-powerful test is unique up to null sets;
thus \(\delta=\mathbf 1_{E_A}\) almost surely, where
\(E_A:=\{Z_A\ge c_A\}\). If \(\Pone(S<A)>0\), then on \(\{S<A\}\) the
conditional increment \(Z_A-Z_S\) is Gaussian with strictly positive variance
\(A-S\), so
\(0<\Pone(E_A\mid\mathcal G_S)<1\). Since \(\delta\) is
\(\mathcal G_S\)-measurable and \(\delta=\mathbf 1_{E_A}\), this would imply
\(\Pone(E_A\mid\mathcal G_S)=\delta\in\{0,1\}\), a contradiction. Thus
\(S=A\) almost surely at the critical point and the post-certification value is
zero.
\end{proof}

For example, with \(\alpha=0.05\) and power \(0.90\), \[
A_{\rm crit}=8.5638473507\ldots,
\qquad
\mathcal I_{\rm crit}=4.2819236753\ldots\ \text{nats}.
\] If $\mathcal I$ denotes the KL information available to an arbitrary binary decision experiment, the generic data-processing bound is
$\mathcal I\ge\operatorname{kl}(0.9,0.05)=2.3762054\ldots$ nats, which is below the canonical Gaussian requirement.

A constructive fixed-time lower bound for
\(A>A_{\rm crit}\) is \[
c_{\alpha,\beta}(A)
\ge
\sup_{A_{\rm crit}\le t<A}
(A-t)\,\Phi\!\left(\sqrt t-z_{1-\alpha}\right).
\] It tests once at time \(t\). The optimal stopping rule can only
improve on this benchmark.

The zero residual value at the critical horizon is the endpoint implication of the most-powerful Gaussian test. In the discrete Gaussian experiment, Clerico et al.~\cite[Sections 6 and 8]{clericoetal2026} likewise show that the hard-deadline power-optimal test cannot reject before its deadline. The next result quantifies the residual value obtainable when the information horizon exceeds that endpoint by a vanishing amount.

\subsection{Near-critical capacity}\label{near-critical-capacity}

Assume $0<\alpha<1-\beta<1$. For the remainder of the paper write $A_*:=A_{\rm crit}(\alpha,\beta)$. Let
\[
\varphi_N(x):=(2\pi)^{-1/2}e^{-x^2/2}
\]
denote the standard normal density.  The fixed-time rule at $A_*$ earns $(1-\beta)(A-A_*)$ when it is embedded in a longer horizon.  Sequential certification has a different local scale.

\begin{theorem}[Sharp near-critical capacity]\label{thm:near-critical-capacity}
For the canonical Brownian likelihood experiment,
\[
c_{\alpha,\beta}(A_*)=0,
\qquad
c_{\alpha,\beta}(A_*+\varepsilon)
\sim
\frac{\mathcal K_{\alpha,\beta}}{\log(1/\varepsilon)},
\qquad \varepsilon\downarrow0,
\]
where
\begin{equation}\label{eq:sharp-capacity-constant}
\mathcal K_{\alpha,\beta}
:=
\frac{A_*}{2}
\left(
[1+z_{1-\beta}^2](1-\beta)
+z_{1-\beta}\varphi_N(z_{1-\beta})
\right).
\end{equation}
The upper bound holds over randomized stopping rules.  A binary conditional-error crossing rule attains the limit.  Consequently
\[
\lim_{\varepsilon\downarrow0}
\frac{c_{\alpha,\beta}(A_*+\varepsilon)}{\varepsilon}=+\infty.
\]
\end{theorem}

The constant has a direct Gaussian interpretation. For a standard normal
dummy variable $N_{\rm G}$,
\[
\mathcal K_{\alpha,\beta}
=\frac{A_*}{2}\mathbb E[(z_{1-\beta}-N_{\rm G})_+^2].
\]
The positive part measures the squared terminal likelihood margin on paths
that can support an earlier certificate; the expectation averages that
margin across the alternative distribution. The constant measures information time in the canonical experiment; cash value follows from the payoff maps of Section~\ref{sharp-economic-capacity}.

For $\alpha=0.05$ and power $0.90$,
\[
\mathcal K_{0.05,0.10}=11.1460491583\ldots.
\]
Appendix~\ref{app:sharp-capacity-proof} proves the universal upper bound and the attaining construction.  The fixed-time benchmark is $(1-\beta)\varepsilon$, so the adaptive-to-fixed-time value ratio diverges at rate
\[
\frac{1}{\varepsilon\log(1/\varepsilon)}.
\]

\begin{table}[ht]
\centering\small
\begin{tabular}{@{}rrrr@{}}\toprule
$\varepsilon$ & Interim fraction $t/A$ & Two-look value & Fixed-time value\\\midrule
$10^{-1}$ & $0.50$ & $1.585621$ & $9\times10^{-2}$\\
$10^{-2}$ & $0.80$ & $1.001639$ & $9\times10^{-3}$\\
$10^{-3}$ & $0.80$ & $0.788752$ & $9\times10^{-4}$\\
$10^{-4}$ & $0.80$ & $0.610641$ & $9\times10^{-5}$\\
$10^{-5}$ & $0.80$ & $0.466293$ & $9\times10^{-6}$\\
$10^{-6}$ & $0.90$ & $0.386198$ & $9\times10^{-7}$\\\bottomrule
\end{tabular}
\caption{Residual information-time value at $\alpha=0.05$, $\beta=0.10$ and $A=A_*+\varepsilon$.  The two-look column reports feasible policies selected from five interim fractions; it is a lower bound on unrestricted capacity.  The fixed-time column tests at $A_*$.}
\label{tab:near-critical-two-look}
\end{table}

The sharp law concerns a deterministic information horizon and allows the preterminal monitoring date to move with $\varepsilon$.  Fixed monitoring grids and positive implementation leads define different local problems; Appendix~\ref{app:monitoring-scope} records the corresponding boundary shift.

\begin{figure}[ht]
\centering
\includegraphics[width=0.94\linewidth]{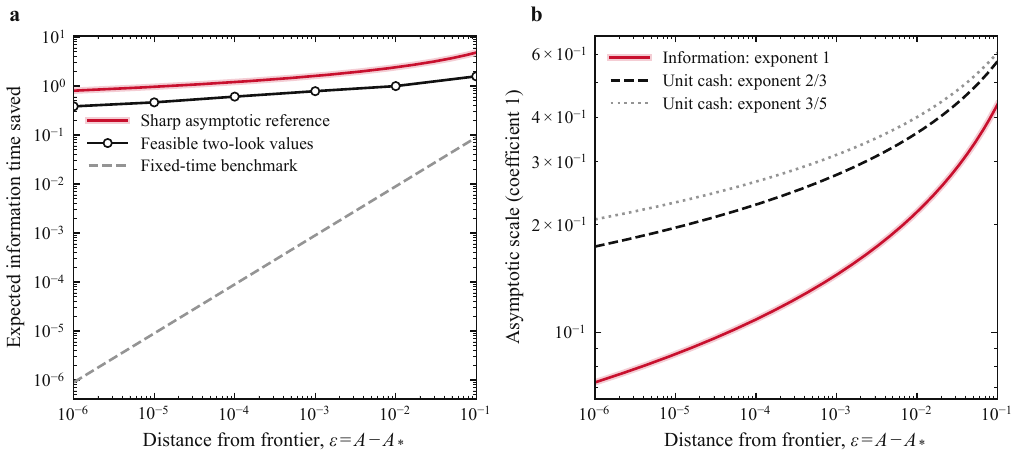}
\caption{Near-frontier information and cash-value scales. Left: the sharp
asymptotic reference $\mathcal K_{\alpha,\beta}/\log(1/\varepsilon)$,
feasible two-look values from Table~\ref{tab:near-critical-two-look}, and
the fixed-time benchmark, at $\alpha=0.05$ and $\beta=0.10$.
The reference is not a computed finite-horizon optimum. Right:
logarithmic scales with coefficients normalized to one; the cash
exponents require the uniform reward assumptions in
Section~\ref{sharp-economic-capacity}. The exponent $(p+1)/(2p+1)$ is $2/3$ for a transverse crossing ($p=1$) and $3/5$ for a quadratic tangency ($p=2$).}
\label{fig:nearcritical}
\end{figure}

\subsection{A general information lower bound}\label{a-generic-fixed-power-converse}

\begin{proposition}[Information lower bound for reliable certification]
Suppose $A>0$ and $1-\beta\ge\alpha$. Whenever the reliable constraint set is nonempty, its capacity satisfies
\[
\frac{A-c_{\alpha,\beta}(A)}{2}
\ge \operatorname{kl}(1-\beta,\alpha).
\]
Hence
\[
\frac{\mathsf{CAC}_{\alpha,\beta}(A)}{V^{PI}}
\le
1-\frac{\operatorname{kl}(1-\beta,\alpha)}{\mathcal I_{\rm life}}.
\]
Feasibility requires the right-hand side to be nonnegative; otherwise no rule satisfies the declared power and type-I targets.
\end{proposition}

\begin{proof}
Take any feasible rule and write \(R=\Eone[(A-S)\delta]\),
\(p=\Pone(\delta=1)\), and \(q=\Pzero(\delta=1)\). Setting \(S=A\) on
nondeployment gives
\[
R=A-\Eone S.
\]
The stopped likelihood ratio is \(e^{Z_S}\), so data processing from
the stopped path to the binary decision \(\delta\) yields
\[
\frac12\Eone S
=\KL(\Pone^S\Vert\Pzero^S)
\ge \operatorname{kl}(p,q)
\ge \operatorname{kl}(1-\beta,\alpha),
\]
where the last inequality uses \(p\ge1-\beta\ge\alpha\ge q\) and
the monotonicity of binary relative entropy on this region. Hence every
feasible rule satisfies
\(R\le A-2\operatorname{kl}(1-\beta,\alpha)\). Taking the supremum
over feasible rules proves the result.
\end{proof}

The exact Gaussian threshold gives the feasibility boundary in the canonical
experiment. The KL inequality extends the information-budget argument beyond
that model and also bounds the economic value consumed by certification.

\section{Economic boundary and the value of information}\label{economic-nulls-and-information-geometry}

Value-of-information methods evaluate evidence through the decision it changes \cite{jaimungalshi2024,lehrerwang2024,anand2010,zhong2026}. Let $b_t<\theta_t$ be a predictable economic-boundary loading and impose the joint-likelihood and Girsanov assumptions of Section~\ref{probabilistic-setup} for the gap $\theta-b$. The initial observation and predictor add no separate likelihood factor. For the corresponding laws $\mathbb P_\theta,\mathbb P_b$,
\[
A_t^{\theta,b}=\sigma^{-2}\int_0^t(\theta_u-b_u)^2X_u^2\,du,
\qquad \KL(\mathbb P_\theta^t\Vert\mathbb P_b^t)=\tfrac12\mathbb E_\theta A_t^{\theta,b}.
\]
Remaining flat during learning forgoes $\theta_t^2X_t^2dt/(2\gamma)$, so its price per KL nat is $\Pi_t^{\rm flat}=(\sigma^2/\gamma)[\theta_t/(\theta_t-b_t)]^2$, which diverges near the boundary. The statistical experiment itself remains regular.

\begin{theorem}[Value relative to the economic boundary]\label{thm:boundary-relative-value}
At a fixed time, for $u(q;\theta)=q\theta X-\gamma q^2/2$ with $\gamma>0$, the optima are $q_\theta=\theta X/\gamma$ and $q_b=bX/\gamma$. Their regret under the true loading is
\[
[u(q_\theta;\theta)-u(q_b;\theta)]dt
=\frac{(\theta-b)^2X^2}{2\gamma}\,dt
=\frac{\sigma^2}{2\gamma}\,dA^{\theta,b}.
\]
The boundary-relative price per KL nat is therefore $\Pi^{\rm relative}=\sigma^2/\gamma$, independent of the gap.
\end{theorem}
\begin{proof} Complete the square in $u$ and divide the regret rate by the Gaussian drift-change KL rate $(\theta-b)^2X^2/(2\sigma^2)$.\end{proof}
The conversion depends on the admissible pre-certification policy: $\Pi^{\rm flat}$ applies to zero exposure, while the analysis below uses the boundary-relative benchmark.

\subsection{Several assets: value per unit of information}\label{multivariate-economicinformation-spectrum}
For $dY_t=\mu\,dt+\Sigma^{1/2}dB_t$ and $u(q;\mu)=q^\top\mu-q^\top\Gamma q/2$, assume $\Sigma,\Gamma\in\mathbb S_{++}^d$. The true and boundary path laws share $\Sigma$; put $\Delta=\mu-b$ and $\Psi=\Sigma^{1/2}\Gamma^{-1}\Sigma^{1/2}$.
\begin{theorem}[Multivariate value--information bound]
For the boundary position $q_b=\Gamma^{-1}b$ relative to $q_\mu=\Gamma^{-1}\mu$, the regret and KL rates are $dR=\Delta^\top\Gamma^{-1}\Delta\,dt/2$ and $d\mathcal I=\Delta^\top\Sigma^{-1}\Delta\,dt/2$. For $\Delta\ne0$,
\[
\frac{dR}{d\mathcal I}
=\frac{\Delta^\top\Gamma^{-1}\Delta}{\Delta^\top\Sigma^{-1}\Delta}
\in[\lambda_{\min}(\Psi),\lambda_{\max}(\Psi)].
\]
The ratio is direction-independent exactly when $\Psi=cI_d$ for $c>0$, equivalently $\Gamma=c^{-1}\Sigma$; then it equals $c$.
\end{theorem}
\begin{proof} Complete the square for regret and use Gaussian Girsanov for KL. With $x=\Sigma^{-1/2}\Delta$, the ratio is the Rayleigh quotient $x^\top\Psi x/(x^\top x)$; it is constant in every direction exactly when $\Psi=cI_d$.\end{proof}

\section{Alpha Survival Frontiers}\label{alpha-survival-frontiers}

Alpha decay has been studied both as a post-discovery empirical phenomenon and as an input to trading and strategy-durability models \cite{penasse2022,masmith2025,varma2026,zulfiqar2026,alexanderfabozzi2026}. The object here is the amount of statistical information that remains before the signal reaches its economic boundary.  For the local terminal asymptotics in this section, take the drift and covariance paths to be deterministic; the same formulas hold pathwise on any realization for which the stated expansions are valid.  Let the observed $d$-dimensional return channel, under the true market model and an economic-boundary model, differ by drift gap
\[
\Delta_t=\mu_t-b_t\in\mathbb R^d,
\]
with instantaneous noise covariance $\Sigma_t\in\mathbb S_{++}^d$. Throughout this section, a superscript $(b)$ on an information quantity means that KL divergence is measured relative to the economic-boundary model $b=(b_t)$. Let \(T\) be
the end of the opportunity's economic lifetime. Assume that as \(t\uparrow T\), for some
\(p>0\) and nonzero \(c\in\mathbb R^d\), \[
\Delta_t=c(T-t)^p+o((T-t)^p),
\qquad
\Sigma_t\to\Sigma_T\succ0.
\] Define remaining KL information from time \(t\) until death by \[
\mathcal I_{\rm rem}^{(b)}(t)
:=\frac12\int_t^T
\Delta_s^\top\Sigma_s^{-1}\Delta_s\,ds.
\]

\begin{theorem}[Remaining information near the economic endpoint]\label{thm:terminal-collapse}
Under the assumptions above,
\[
\mathcal I_{\rm rem}^{(b)}(t)
=\frac{c^\top\Sigma_T^{-1}c}{2(2p+1)}(T-t)^{2p+1}
+o((T-t)^{2p+1}).
\]
In particular, for a transverse crossing $p=1$,
\[
\mathcal I_{\rm rem}^{(b)}(t)
\sim \frac{c^\top\Sigma_T^{-1}c}{6}(T-t)^3.
\]
\end{theorem}

\begin{proof}
Put \(u=T-s\). By the assumed expansion and continuity
of matrix inversion on the positive-definite cone, \[
\Delta_{T-u}^\top\Sigma_{T-u}^{-1}\Delta_{T-u}
=u^{2p}\,c^\top\Sigma_T^{-1}c+o(u^{2p}).
\] Integrating from \(u=0\) to \(u=T-t\) gives \[
\mathcal I_{\rm rem}^{(b)}(t)
=\frac12 c^\top\Sigma_T^{-1}c
\frac{(T-t)^{2p+1}}{2p+1}
+o((T-t)^{2p+1}).
\] \end{proof}

\paragraph{Economic value and the latest start.}\label{economic-value-has-the-same-exponent}\label{latest-feasible-start}
Put $V_{\rm rem}^{\mu:b}(t):=\tfrac12\int_t^T\Delta_s^\top\Gamma_s^{-1}\Delta_s\,ds$. For quadratic economic curvature $\Gamma_t\to\Gamma_T\succ0$, completing the square and integrating as above gives $V_{\rm rem}^{\mu:b}(t)\sim[c^\top\Gamma_T^{-1}c/(2(2p+1))](T-t)^{2p+1}$. Thus $V_{\rm rem}^{\mu:b}/\mathcal I_{\rm rem}^{(b)}\to(c^\top\Gamma_T^{-1}c)/(c^\top\Sigma_T^{-1}c)$. If $\Delta_{T-u}=cu^p$ and $\Sigma$ is constant on the terminal interval, the information law is exact. With $\mathcal I_*=[z_{1-\alpha}+z_{1-\beta}]^2/2$, a newly started experiment needs at least
\[
\ell_*=[2(2p+1)\mathcal I_* /(c^\top\Sigma^{-1}c)]^{1/(2p+1)}
\]
units of remaining time, provided this boundary lies in the stated interval. For a smooth crossing this gives the local asymptotic scale of the latest-start boundary; the global boundary depends on the full path geometry. An experiment already running also carries its likelihood state and remaining error budget.

\paragraph{Positive economic hurdle.}\label{exponential-decay-with-a-positive-economic-hurdle}\label{cubic-near-hurdle-collapse}
For $\mu_t=\mu_0e^{-\lambda_{\rm dec}t}$, $0<b<\mu_0$, $\lambda_{\rm dec}>0$, death occurs at $T=\lambda_{\rm dec}^{-1}\log(\mu_0/b)$. Writing $r_b=b/\mu_0$, direct integration against boundary drift $b$ gives
\[
\mathcal I_{\rm life}^{(b)}=\frac{\mu_0^2}{2\sigma^2\lambda_{\rm dec}}
\left[\tfrac12-2r_b+\tfrac32r_b^2+r_b^2\log(1/r_b)\right]
\sim\frac{\mu_0^2(1-r_b)^3}{6\sigma^2\lambda_{\rm dec}}
\quad(r_b\uparrow1).
\]
The bracket expands as $(1-r_b)^3/3+O((1-r_b)^4)$. Locally $\mu_{T-u}-b\sim b\lambda_{\rm dec}u$, so Theorem~\ref{thm:terminal-collapse} also gives $\mathcal I_{\rm rem}^{(b)}(T-\ell)\sim b^2\lambda_{\rm dec}^2\ell^3/(6\sigma^2)$.

\subsection{Zero-hurdle exponential model: the Alpha Survival
Frontier}\label{zero-hurdle-exponential-model-the-alpha-survival-frontier}

If the economic boundary is zero and \(\mu_t=\mu_0e^{-{\lambda_{{\rm dec}}} t}\) for all
\(t\ge0\), the calendar horizon is unbounded while total lifetime information
remains finite: \[
\mathcal I_{\rm life}
=\frac1{2\sigma^2}\int_0^\infty\mu_0^2e^{-2{\lambda_{{\rm dec}}} t}dt
=\frac{\mu_0^2}{4\sigma^2{\lambda_{{\rm dec}}}}.
\] Let \[
\mathsf S:=\frac{\mu_0}{\sigma},
\qquad
h_{1/2}:=\frac{\log2}{{\lambda_{{\rm dec}}}}
\] be the instantaneous Sharpe (with time units fixed consistently) and
alpha half-life. Then \[
\mathcal I_{\rm life}
=\frac{\mathsf S^2h_{1/2}}{4\log2}.
\] Reliable certification is feasible only if \[
\frac{\mathsf S^2h_{1/2}}{4\log2}
\ge
\frac12[z_{1-\alpha}+z_{1-\beta}]^2,
\] or equivalently \[
h_{1/2}\ge h_{\min}(\mathsf S;\alpha,\beta)
:=\frac{2\log2}{\mathsf S^2}
[z_{1-\alpha}+z_{1-\beta}]^2.
\] We refer to this boundary as the canonical \emph{Alpha Survival Frontier}.

For \(\alpha=0.05\), power \(1-\beta=0.90\):

\begin{center}
\begin{tabular}{@{}rr@{}}
\toprule
instantaneous Sharpe $\mathsf S$ & minimum half-life \\
\midrule
1.0 & 11.872 \\
1.5 & 5.276 \\
2.0 & 2.968 \\
3.0 & 1.319 \\
\bottomrule
\end{tabular}
\end{center}

These thresholds use the calendar-time unit implicit in the scaling of $\mathsf S$ and are implied by the stated exponential-decay model.

\subsection{Multiple testing and search breadth}\label{multiple-search-tax}

Large-scale strategy search creates the familiar multiplicity and selection problem in empirical finance \cite{harveyliu2016,baileyloprado2014,baileyborwein2017}. Suppose $M\in\mathbb N$ candidate signals are screened under familywise
false-deployment budget $\alpha_F\in(0,1)$. Under Bonferroni, assign \[
\alpha_M=\frac{\alpha_F}{M}.
\] The canonical information threshold for one selected signal is then
\[
\mathcal I_{\rm crit}^{(M)}
=\frac12\left[
 z_{1-\alpha_F/M}+z_{1-\beta}
\right]^2.
\] Under a Bonferroni allocation, this is the exact per-test Gaussian threshold. Bonferroni remains valid under arbitrary dependence; additional dependence information can support less conservative allocations.

The zero-hurdle exponential survival frontier becomes \[
 h_{\min}^{(M)}(\mathsf S)
 =\frac{2\log2}{\mathsf S^2}
 \left[z_{1-\alpha_F/M}+z_{1-\beta}\right]^2.
\] At \(\alpha_F=0.05\), power \(0.90\):

\begin{center}
\begin{tabular}{@{}rrrrr@{}}
\toprule
\(M\) & $h_{\min}$, $\mathsf S=1$ & $\mathsf S=1.5$ & $\mathsf S=2$ & $\mathsf S=3$ \\
\midrule
1 & 11.872 & 5.276 & 2.968 & 1.319 \\
10 & 20.627 & 9.168 & 5.157 & 2.292 \\
100 & 28.979 & 12.880 & 7.245 & 3.220 \\
1000 & 37.085 & 16.482 & 9.271 & 4.121 \\
\bottomrule
\end{tabular}
\end{center}

\begin{figure}[t]
\centering
\includegraphics[width=0.92\linewidth]{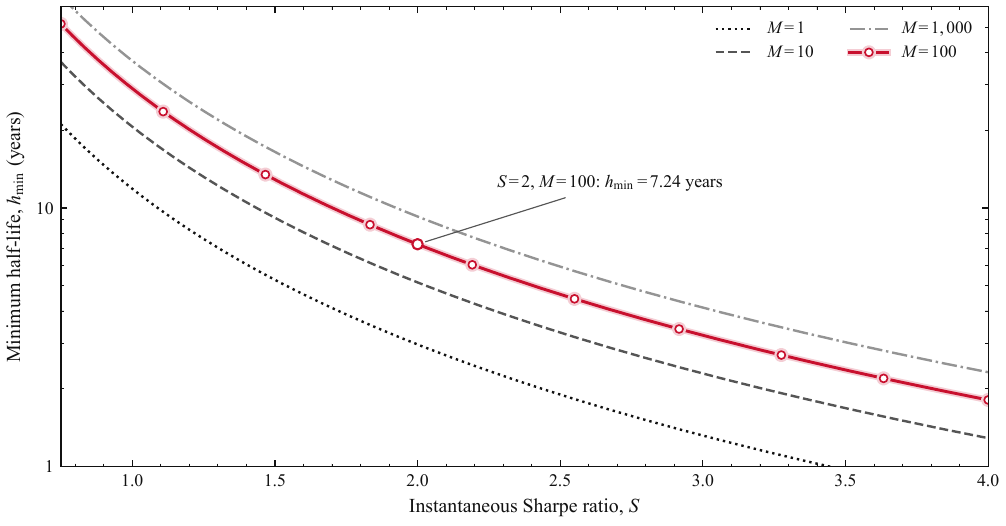}
\caption{Canonical Alpha Survival Frontier for familywise false-deployment level $0.05$ and power $0.90$. Each curve gives the minimum half-life required at a given instantaneous Sharpe for a Bonferroni search over $M$ candidates. The highlighted point is $\mathsf S=2$, $M=100$, for which the minimum half-life is $7.245$ years. The pale-red outline is a visual emphasis of the $M=100$ curve.}
\label{fig:survival-frontier}
\end{figure}

\paragraph{Search breadth and unequal budgets.}\label{maximum-search-breadth}\label{unequal-error-budget}
Inverting the one-candidate power envelope gives the minimum level
$\alpha_{\min}(\mathcal I,\beta)=\bar\Phi(\sqrt{2\mathcal I}-z_{1-\beta})$.
Hence $M_{\max}^{\rm Bonf}=\lfloor\alpha_F/\alpha_{\min}\rfloor$ (zero means no feasible candidate). Independent tests permit the Sidak inversion $M_{\max}^{\rm Sidak}=\lfloor\log(1-\alpha_F)/\log(1-\alpha_{\min})\rfloor$; Bonferroni needs no independence.
\begin{remark}[Heterogeneous Bonferroni feasibility]\label{thm:heterogeneous-bonferroni}
Within the separable union-bound architecture, candidate $j$ with information $\mathcal I_j\ge0$ and power $1-\beta_j$ needs $\alpha_j^*=\bar\Phi(\sqrt{2\mathcal I_j}-z_{1-\beta_j})$. All targets are feasible exactly when $\sum_j\alpha_j^*\le\alpha_F$: necessity follows from the marginal envelopes and sufficiency from the union bound. With equal candidate values, sorting these costs and taking the longest feasible prefix maximizes cardinality; heterogeneous values give a $0$--$1$ knapsack constraint.
\end{remark}
For fixed $\alpha_F,\beta$, Gaussian tail inversion gives $\mathcal I_{\rm crit}^{(M)}\sim\log M$ as $M\to\infty$.\label{large-search-asymptotics}

\subsection{A dimensionless feasibility ratio}\label{certifiability-ratio}
For a declared model, boundary and error allocation, define $\mathrm{CR}=\mathcal I_{\rm life}/\mathcal I_{\rm crit}(\alpha,\beta)$, using the declared per-test level under multiplicity. In the canonical experiment, $\mathrm{CR}<1$ is infeasible, $\mathrm{CR}=1$ is terminal feasibility, and $\mathrm{CR}>1$ permits reliable certification. This ratio is a design-feasibility diagnostic; statistical evidence and deployment decisions depend on the observed likelihood and the selected stopping rule. An ongoing experiment must retain its likelihood state $Z_s$ and original error allocation, as in Appendix~\ref{constrained-stopping-representation}. The terminal envelope $\pi_\alpha(\mathrm{CR}A_*)$ bounds the power of every rule satisfying the declared false-deployment bound; a stopping policy that preserves economic value can have lower power.

\section{Certification timing and post-certification cash value}\label{sharp-economic-capacity}

The capacity $c_{\alpha,\beta}(A)$ measures information time left at certification. A trading payoff is generally a function of that residual clock. The sharp near-critical law from Section~\ref{near-critical-capacity} transfers to payoff maps through their local dependence on residual information.

For $\rho_{\rm R}>0$ define the canonical residual-moment function
\begin{equation}\label{eq:residual-moment-function}
\mathcal M_{\alpha,\beta}(\rho_{\rm R})
:=A_*^{\rho_{\rm R}}
\mathbb E\!\left[(N_{\rm G}+z_{1-\beta})_+^{2\rho_{\rm R}}\right],
\qquad N_{\rm G}\sim N(0,1),
\end{equation}
where $A_*=A_{\rm crit}(\alpha,\beta)$ as in Section~\ref{near-critical-capacity}.
The reward extension is taken along a family of canonical experiments with
$A=A_*+\varepsilon\downarrow A_*$. Let $\mathcal R_A:[0,A]\to[0,\infty)$
be deterministic measurable residual rewards satisfying $\mathcal R_A(0)=0$,
uniformly bounded for $A$ in a right neighborhood of $A_*$. Assume that,
for some $a_0>0$ and $\rho_{\rm R}>0$,
\[
\mathcal R_A(u)/u^{\rho_{\rm R}}\longrightarrow a_0
\quad\text{uniformly as }A\downarrow A_*\text{ and }u\downarrow0.
\]
Appendix~\ref{app:general-residual-rewards} then gives
\begin{equation}\label{eq:general-residual-reward}
\sup_{(S,\delta)}\Eone\!\left[\mathcal R_A(A-S)\delta\right]
\sim
\frac{a_0\,\mathcal M_{\alpha,\beta}(\rho_{\rm R})}{[2\log(1/\varepsilon)]^{\rho_{\rm R}}},
\qquad \varepsilon\downarrow0,
\end{equation}
where the supremum is over the same reliable certification rules as in Section~\ref{reliable-deployment-and-certified-alpha-capacity}.

\subsection{Bounded exposure near economic death}\label{bounded-exposure-economic-death}

Use the scalar boundary-relative drift gap from Section~\ref{alpha-survival-frontiers}.  Suppose that, as $u\downarrow0$,
\[
\Delta_{T-u}=c\,u^p+o(u^p),
\qquad
\sigma_{T-u}\to\sigma_T>0,
\qquad c>0,
\]
with the same crossing order $p>0$ as in Theorem~\ref{thm:terminal-collapse}.  A unit position is opened only after certification and then held through the economic endpoint $T$.  Its expected boundary-relative cash reward from time $t$ is
\[
V_{\rm unit}(t):=\int_t^T\Delta_s\,ds.
\]
Theorem~\ref{thm:terminal-collapse} gives the residual information-time scale
\[
2\mathcal I_{\rm rem}^{(b)}(T-u)
\sim
\frac{c^2}{\sigma_T^2(2p+1)}u^{2p+1},
\]
while
\[
V_{\rm unit}(T-u)
\sim\frac{c}{p+1}u^{p+1}.
\]
Hence
\begin{equation}\label{eq:unit-reward-vs-information}
V_{\rm unit}
\sim K_p^{\rm unit}\,u_{\rm info}^{\rho_p},
\qquad
\rho_p:=\frac{p+1}{2p+1},
\qquad
K_p^{\rm unit}:=\frac{c}{p+1}
\left(\frac{\sigma_T^2(2p+1)}{c^2}\right)^{\rho_p},
\end{equation}
when $u_{\rm info}:=2\mathcal I_{\rm rem}^{(b)}$ denotes residual boundary-relative information time.
For the near-critical statement, consider a family with total information
$A\downarrow A_*$ whose terminal expansions are uniform and whose residual
reward maps satisfy the hypotheses above. The constants $c$ and $\sigma_T$
denote positive limiting terminal coefficients. Combining
\eqref{eq:general-residual-reward} and \eqref{eq:unit-reward-vs-information}
gives the sharp cash-value capacity
\begin{equation}\label{eq:cash-capacity-sharp}
\mathsf{CAC}^{\rm cash}_{\alpha,\beta}(A_*+\varepsilon)
\sim
\frac{K_p^{\rm unit}\,\mathcal M_{\alpha,\beta}(\rho_p)}
{[2\log(1/\varepsilon)]^{\rho_p}}.
\end{equation}
For a transverse crossing, $p=1$ and $\rho_p=2/3$, so the attainable post-certification cash value is of order
\[
[\log(1/\varepsilon)]^{-2/3}.
\]
The exponent differs from the information-time exponent because a bounded position accumulates $\int\Delta_t\,dt$, whereas the likelihood clock accumulates $\int\Delta_t^2/\sigma_t^2\,dt$.

For the exponential gap
\[
\Delta_t=\mu_0e^{-{\lambda_{{\rm dec}}} t}-b,
\qquad
T={\lambda_{{\rm dec}}}^{-1}\log(\mu_0/b),
\]
Section~\ref{exponential-decay-with-a-positive-economic-hurdle} has $p=1$ and $c=b{\lambda_{{\rm dec}}}$.  Formula \eqref{eq:cash-capacity-sharp} applies directly.

\subsection{A fixed deployment fee}\label{fixed-transaction-fees}

A fixed round-trip fee shifts the last economically useful certification time away from the zero-edge endpoint.  Consider the local linear gap
\[
\Delta_t=c(T-t),\qquad 0\le t\le T,
\]
and a fixed fee $F_{\rm fee}\in(0,cT^2/2)$.  The last entry time with nonnegative expected net cash reward is
\[
T_{\rm fee}=T-\ell_{\rm fee},
\qquad
\ell_{\rm fee}=\sqrt{\frac{2F_{\rm fee}}{c}}.
\]
Take constant volatility $\sigma>0$ and use the information horizon
ending at the last admissible entry date:
\[
A_{\rm fee}:=\frac1{\sigma^2}\int_0^{T_{\rm fee}}c^2(T-t)^2\,dt
=\frac{c^2}{3\sigma^2}(T^3-\ell_{\rm fee}^3).
\]
The near-critical fee law below uses $A_{\rm fee}=A_*+\varepsilon$.
If $u_{\rm fee}$ is the information time remaining to $T_{\rm fee}$, the exact net reward from certifying at that state is
\begin{equation}\label{eq:fixed-fee-reward}
\mathcal R_{\rm fee}(u_{\rm fee})
=
\frac{c}{2}
\left[
\left(\ell_{\rm fee}^3+\frac{3\sigma^2u_{\rm fee}}{c^2}\right)^{2/3}
-\ell_{\rm fee}^2
\right].
\end{equation}
Thus
\[
\mathcal R_{\rm fee}(u_{\rm fee})
\sim\frac{\sigma^2}{c\ell_{\rm fee}}u_{\rm fee},
\qquad u_{\rm fee}\downarrow0,
\]
and the sharp net-return capacity returns to the inverse-logarithmic scale:
\begin{equation}\label{eq:fixed-fee-capacity}
\mathsf{CAC}^{\rm fee}_{\alpha,\beta}(A_*+\varepsilon)
\sim
\frac{\sigma^2\mathcal K_{\alpha,\beta}}
{c\ell_{\rm fee}\log(1/\varepsilon)}.
\end{equation}
The fee changes the local asymptotic law because the relevant deployment boundary is reached while the underlying drift gap is still strictly positive.

\subsection{Cash value beyond lifetime information}\label{cr-not-economic-ranking}

The ratio $CR$ identifies statistical feasibility in the canonical simple-hypothesis experiment.  Cash value also depends on how the reward is distributed along the information clock.  In particular, two opportunities can share the same current edge, volatility, calendar horizon, total future cash value, lifetime information and the same $CR$, while having different local crossing orders at economic death.  Equation~\eqref{eq:cash-capacity-sharp} then gives different certification-constrained cash values.  Appendix~\ref{app:cash-capacity-counterexample} gives an explicit pair with identical $CR$ and different asymptotic exponents.

\section{Unknown decay and robust certification}\label{unknown-decay-robust-certification}

The preceding sections condition on a specified boundary-relative drift profile.  We now require a single observable policy to satisfy the same false-deployment and power targets over a declared family of possible decay paths.

\subsection{Ordered families}\label{ordered-unknown-decay}

Fix a deterministic horizon $T$ and let $\{\Delta_\vartheta:\vartheta\in\Theta\}$ be continuous scalar drift gaps satisfying
\[
\Delta_\vartheta(t)\ge \underline\Delta(t)>0,
\qquad 0\le t<T,
\qquad \underline\Delta(T)\ge0,
\]
where the least profile $\underline\Delta$ belongs to the family.  Under alternative $\vartheta$,
\[
dY_t=\Delta_\vartheta(t)\,dt+\sigma\,dB_t,
\]
while the boundary law has zero drift.  Define the least-profile information clock and matched log-likelihood by
\[
\underline A_t
:=\frac1{\sigma^2}\int_0^t\underline\Delta(u)^2\,du,
\qquad
\underline A:=\underline A_T,
\]
\[
\underline Z_t
:=\frac1{\sigma^2}\int_0^t\underline\Delta(u)\,dY_u
-\frac12\underline A_t.
\]

For $0<\alpha<1-\beta<1$, define the robust residual-information capacity by
\[
c^{\rm rob}_{\alpha,\beta}(\underline A)
:=\sup_{(S,\delta)}\inf_{\vartheta\in\Theta}
\mathbb E_\vartheta[(\underline A-S)\delta],
\]
where $S\le\underline A$ is a stopping time on the deterministic least-profile information clock, $0\le\delta\le1$ is measurable at $S$, and the common policy satisfies
$\Ezero\delta\le\alpha$ and $\inf_\vartheta\mathbb E_\vartheta\delta\ge1-\beta$.
The declared family is suppressed in the capacity notation.

\begin{theorem}[Ordered-family robust frontier]\label{thm:ordered-decay-frontier}
Among all terminal decisions of size at most $\alpha$,
\[
\sup_{\delta:\,\Ezero\delta\le\alpha}
\inf_{\vartheta\in\Theta}\mathbb E_\vartheta\delta
=
\Phi\!\left(\sqrt{\underline A}-z_{1-\alpha}\right).
\]
Consequently a common certificate with power at least $1-\beta$ for every $\vartheta\in\Theta$ exists if and only if
\[
\underline A\ge A_*.
\]
If $\underline A=A_*+\varepsilon$, the robust residual-information capacity satisfies
\[
c^{\rm rob}_{\alpha,\beta}(\underline A)
\sim
\frac{\mathcal K_{\alpha,\beta}}{\log(1/\varepsilon)}.
\]
For a common bounded, nonnegative, nondecreasing reward $\mathcal R(u)$ with $\mathcal R(0)=0$, assume $\mathcal R(u)\sim a_0u^{\rho_{\rm R}}$ with the uniformity required by the reward-transfer result. Then
$\sup_{(S,\delta)}\inf_\vartheta\mathbb E_\vartheta[\mathcal R(\underline A-S)\delta]$
obeys \eqref{eq:general-residual-reward} with the least-profile constant. The same conclusion holds for member-specific payoff maps when every member's reward at each calendar time dominates this least-profile reward.
\end{theorem}

The matched statistic is ordered pathwise under a common Brownian coupling, and the least profile belongs to the declared family. It determines terminal worst-case power. The conditional-error crossing rule stops no later for larger profiles; under the stated reward-ordering assumption, this also establishes the sharp robust-reward lower bound. The least member supplies the matching upper bound.

For the exponential family
\[
\Delta_{\theta,{\lambda_{{\rm dec}}}}(t)=\theta e^{-{\lambda_{{\rm dec}}} t}-b,
\]
with $\theta\in[\theta_-,\theta_+]$, ${\lambda_{{\rm dec}}}\in[{\lambda_{{\rm dec},-}},{\lambda_{{\rm dec},+}}]$ and $\theta_->b>0$, the common least profile up to the earliest economic death is
\[
\underline\Delta(t)=\theta_-e^{-{\lambda_{{\rm dec},+}}t}-b,
\qquad
T={\lambda_{{\rm dec},+}}^{-1}\log(\theta_-/b).
\]
The policy is constructed from $\underline\Delta$ and the observed path, uniformly over the declared parameter rectangle.

\subsection{Nonordered families and least-favorable mixtures}\label{nonordered-unknown-decay}

For nonordered profiles, robust feasibility depends on the geometry of the alternative family beyond memberwise lifetime information. Consider a finite Gaussian experiment with whitened observation $\mathbf Y\in\mathbb R^n$:
\[
\mathbf Y\sim N(0,I_n)\quad(\Pzero),
\qquad
\mathbf Y\sim N(m_j,I_n)\quad(\mathbb P_j),\quad j=1,\ldots,J.
\]
For a terminal decision $0\le\delta\le1$, define
\[
\pi_{\rm rob}
:=\sup_{\Ezero\delta\le\alpha}
\min_{1\le j\le J}\mathbb E_j\delta.
\]
Let
\[
L_j(y)=\exp\!\left(m_j^\top y-\frac12\|m_j\|^2\right),
\]
and let
\[
\mathcal W_J:=\left\{w\in[0,1]^J:\sum_{j=1}^J w_j=1\right\}
\]
be the simplex of mixture weights. For $w=(w_1,\ldots,w_J)\in\mathcal W_J$, set $L_w=\sum_jw_jL_j$.

\begin{proposition}[Least-favorable mixture]\label{prop:least-favorable-mixture}
The exact robust terminal power satisfies
\[
\pi_{\rm rob}
=
\min_{w\in\mathcal W_J}
\sup_{0\le\delta\le1,\,\Ezero\delta\le\alpha}
\Ezero[L_w\delta].
\]
For fixed $w$, the inner problem is the Neyman--Pearson test of $P_0$ against the mixture likelihood $L_w$.  A saddle pair exists.
\end{proposition}

This formulation exposes a strict uncertainty penalty at the individual frontier.  If two distinct nonzero alternatives satisfy
\[
\|m_1\|^2=\|m_2\|^2=A_*,
\]
then each member separately attains power $1-\beta$ at size $\alpha$, but no common level-$\alpha$ decision attains that power under both members.

For the two exponential profiles used in the numerical benchmark, calibrated so that each profile individually lies at $CR=1$ for $\alpha=0.05$ and power $0.90$, the optimal common test has worst-profile power
\[
0.8655488924\ldots.
\]
Scaling both profile norms by the same factor, uniform power $0.90$ is recovered at
\[
CR=1.1268349125\ldots.
\]
Thus the geometry of the alternative family enters robust certification even when every profile has the same oracle lifetime information.

\subsection{Transfer from an estimated model}\label{estimated-model-transfer}

Let $\Pi$ be a common class of admissible policies. For $\pi\in\Pi$, let $\delta_\pi\in\{0,1\}$ denote its deployment decision and let $R_j(\pi)$ denote its integrable reward under alternative $j$. Write
\[
\operatorname{size}(\pi)=P_0(\delta_\pi=1),
\qquad
\operatorname{pow}_j(\pi)=P_j(\delta_\pi=1),
\qquad
\mathcal V_j(\pi)=\mathbb E_j[R_j(\pi)].
\]
Approximate-model quantities carry hats.  Suppose uniformly over $\Pi$ and the declared alternatives that
\[
|\operatorname{size}(\pi)-\widehat{\operatorname{size}}(\pi)|\le\eta_P,
\]
\[
|\operatorname{pow}_j(\pi)-\widehat{\operatorname{pow}}_j(\pi)|\le\eta_P,
\qquad
|\mathcal V_j(\pi)-\widehat{\mathcal V}_j(\pi)|\le\eta_V.
\]
Fix the scalar objective before optimization. Here we use
\[
\mathcal U(\pi):=\inf_j\mathcal V_j(\pi),
\qquad
\widehat{\mathcal U}(\pi):=\inf_j\widehat{\mathcal V}_j(\pi).
\]
A policy optimized for $\widehat{\mathcal U}$ under the tightened approximate constraints
\[
\widehat{\operatorname{size}}\le\alpha-\eta_P,
\qquad
\widehat{\operatorname{pow}}_j\ge1-\beta+\eta_P
\quad\text{for all }j,
\]
satisfies the original reliability contract under the true model. Its worst-case value gap is bounded by the optimization error, $2\eta_V$, and the value of the statistical safety margin, provided the twice-tightened true feasible set is nonempty. This is a bound for the declared scalar objective, not simultaneous optimality for every memberwise reward. Appendix~\ref{app:unknown-decay-proofs} gives the precise statement. For an estimated model, the stated uniform approximation bounds must hold conditionally on the information used to freeze that model, or on an explicitly specified event; they do not follow merely from fitting a parametric family.

\section{Market equilibrium with certification costs}\label{certification-constrained-market-equilibrium}

Information acquisition and arbitrage are endogenous in several market-equilibrium models \cite{grossmanstiglitz1980,basakcroitoru2006,zigrand2004,zigrand2006}. More recent work studies competition among factor investors and the empirical role of post-publication attenuation, arbitrage capital, and implementation frictions \cite{demiguel2025,mcleanpontiff2016,jacobsmuller2020,kaplanski2023,muravyev2025}. The preceding sections take opportunity lifetime as given. We now allow arbitrage activity to shorten it. An explicit crowding map links aggregate research intensity to the remaining information budget, so the CAC threshold becomes an equilibrium constraint.

Retain the reliable-information threshold
\[
\mathcal I_*:=\mathcal I_{\rm crit}(\alpha,\beta)
=\frac12\bigl[z_{1-\alpha}+z_{1-\beta}\bigr]^2.
\]
Let ${v_{\rm KL}}>0$ denote economic value per KL nat. In the canonical scalar model, Theorem~\ref{thm:boundary-relative-value} gives ${v_{\rm KL}}=\sigma^2/\gamma$. Define the optimal reliable CAC economic surplus at lifetime information $\mathcal I\ge0$ by
\[
C^\star(\mathcal I):=\frac {v_{\rm KL}}2\,c_{\alpha,\beta}(2\mathcal I).
\]
The equilibrium arguments below use only the resulting value function and its stated regularity; Appendix~\ref{constrained-stopping-representation} gives the corresponding randomized constrained-stopping representation. For the
equilibrium result we require the natural regularity that \(C^\star\) is
continuous, with
\[
C^\star(\mathcal I)=0\quad(\mathcal I\le \mathcal I_*),
\qquad
C^\star(\mathcal I)>0\quad(\mathcal I>\mathcal I_*).
\]
In the canonical model, the CAC value is strictly increasing above the frontier. Indeed, if
one KL nat of perfect-information value is worth \({v_{\rm KL}}\), then embedding any $\varepsilon$-optimal
rule at \(\mathcal I_1\) in the longer horizon \(\mathcal I_2>\mathcal I_1\ge \mathcal I_*\) adds
\({v_{\rm KL}}(\mathcal I_2-\mathcal I_1)\) on every deployment, whose probability under the alternative law
is at least \(1-\beta\). Letting $\varepsilon\downarrow0$ gives
\[
C^\star(\mathcal I_2)-C^\star(\mathcal I_1)
\ge {v_{\rm KL}}(1-\beta)(\mathcal I_2-\mathcal I_1).
\]
Thus the constructive linear surplus
\[
C_0(\mathcal I)={v_{\rm KL}}(1-\beta)(\mathcal I-\mathcal I_*)_+
\]
is a lower envelope of the optimal CAC value.

Let $\mathcal I_0>0$ be the no-crowding lifetime-information endowment and let
$x\ge0$ denote aggregate arbitrage/research intensity.  The remaining information budget is
\[
\mathcal I=G(x),
\]
where $G:[0,\infty)\to[0,\mathcal I_0]$ is continuous and strictly decreasing,
with $G(0)=\mathcal I_0$. The canonical exponential-crowding benchmark uses a crowding sensitivity $\chi>0$ and is
\[
G(x)=\frac{\mathcal I_0}{1+\chi x},
\]
which arises when ${\lambda_{{\rm dec},0}}>0$ denotes the baseline exponential decay rate and crowding accelerates it from
\({\lambda_{{\rm dec},0}}\) to \({\lambda_{{\rm dec},0}}(1+\chi x)\).

Suppose each unit of active arbitrage capital pays marginal research or
implementation cost \(\psi>0\). In a competitive market with many individually small arbitrageurs,
aggregate post-certification value \(C^\star(G(x))\) is allocated proportionally
across active capacity. Positive-entry equilibrium satisfies
\[
\frac{C^\star(G(x))}{x}=\psi,
\]
or equivalently
\[
C^\star(G(x))=\psi x.
\]

\begin{theorem}[Competitive equilibrium under certification costs]\label{thm:competitive-equilibrium}
Assume \(C^\star\) and \(G\) have the properties above. If \(\mathcal I_0>\mathcal I_*\)
and there is a finite \(x_*\) with \(G(x_*)=\mathcal I_*\), then for every
\(\psi>0\) there is a unique positive equilibrium \(x_\psi\in(0,x_*)\)
satisfying
\[
C^\star(G(x_\psi))=\psi x_\psi.
\]
Under the entry dynamic
\[
\dot x=\eta\{C^\star(G(x))-\psi x\},\qquad \eta>0,
\]
this equilibrium is globally asymptotically stable. Its residual
information \(\mathcal I_\psi=G(x_\psi)\) satisfies
\[
\mathcal I_*<\mathcal I_\psi<\mathcal I_0,
\qquad
\mathcal I_\psi\downarrow \mathcal I_*\quad\text{as }\psi\downarrow0.
\]
If \(\mathcal I_0\le \mathcal I_*\), the no-entry equilibrium \(x=0\) is globally stable.
\end{theorem}

\begin{proof}
Put $F_\psi(x)=C^\star(G(x))-\psi x$. It is continuous and strictly decreasing, with $F_\psi(0)>0$ and $F_\psi(x_*)=-\psi x_*<0$, so it has one root in $(0,x_*)$. Its sign points toward that root on $[0,\infty)$; monotonicity implies the one-sided Lipschitz inequality $(x-y)(F_\psi(x)-F_\psi(y))\le0$, hence uniqueness of the scalar flow and global asymptotic stability. For $0<\psi_1<\psi_2$, evaluating $F_{\psi_1}$ at $x_{\psi_2}$ shows $x_{\psi_1}>x_{\psi_2}$. Its limit as $\psi\downarrow0$ cannot be below $x_*$, since then $C^\star(G(x_\psi))$ would have a positive limit while $\psi x_\psi\to0$. Thus $x_\psi\uparrow x_*$ and $\mathcal I_\psi\downarrow\mathcal I_*$. If $\mathcal I_0\le\mathcal I_*$, the flow is $\dot x=-\eta\psi x$.
\end{proof}

The equilibrium divides opportunities at the certification frontier. Below the frontier, reliable post-certification value is zero. Above the frontier, rents finance entry and crowding reduces the remaining information budget. As research cost falls, the equilibrium residual information approaches the reliable-information threshold.

The lower envelope $C_0$ also yields a useful nonasymptotic bound.  Since $G$ is strictly decreasing, write $G^{-1}$ for its inverse on its range. The optimal equilibrium satisfies
\[
\mathcal I_\psi-\mathcal I_*
\le
\frac{\psi\,G^{-1}(\mathcal I_*)}{{v_{\rm KL}}(1-\beta)}.
\]
For \(G(x)=\mathcal I_0/(1+\chi x)\),
\[
\mathcal I_\psi-\mathcal I_*
\le
\frac{\psi}{\chi {v_{\rm KL}}(1-\beta)}
\left(\frac{\mathcal I_0}{\mathcal I_*}-1\right).
\]
Hence the residual excess information above the CAC frontier vanishes at least
linearly with research cost, without requiring differentiability of the
stopping value at the threshold.

\begin{corollary}[Sharp concentration of the equilibrium gap]\label{cor:near-critical-equilibrium}
Under the hypotheses of Theorem~\ref{thm:competitive-equilibrium}, assume $0<\alpha<1-\beta<1$ and retain
\[
C^\star(\mathcal I)=\frac {v_{\rm KL}}2c_{\alpha,\beta}(2\mathcal I).
\]
Let $x_*=G^{-1}(\mathcal I_*)$.  As $\psi\downarrow0$,
\begin{equation}\label{eq:sharp-equilibrium-gap}
\log\frac1{\mathcal I_\psi-\mathcal I_*}
\sim
\frac{{v_{\rm KL}}\mathcal K_{\alpha,\beta}}{2x_*}\,\frac1\psi.
\end{equation}
Equivalently,
\[
\mathcal I_\psi-\mathcal I_*
=
\exp\!\left[
-\frac{{v_{\rm KL}}\mathcal K_{\alpha,\beta}}{2x_*\psi}
+o(1/\psi)
\right].
\]
\end{corollary}

\begin{proof}
Theorem~\ref{thm:competitive-equilibrium} gives $x_\psi\to x_*>0$ and $\mathcal I_\psi\downarrow\mathcal I_*$. Since $A_*=2\mathcal I_*$, Theorem~\ref{thm:near-critical-capacity} gives
\[
c_{\alpha,\beta}(2\mathcal I_\psi)
\sim
\frac{\mathcal K_{\alpha,\beta}}
{\log\!\left(1/[2(\mathcal I_\psi-\mathcal I_*)]\right)}.
\]
The equilibrium identity $C^\star(\mathcal I_\psi)=\psi x_\psi$ implies
\[
\log\!\left(\frac1{2(\mathcal I_\psi-\mathcal I_*)}\right)
\sim
\frac{{v_{\rm KL}}\mathcal K_{\alpha,\beta}}{2x_*\psi}.
\]
The additive constant $\log2$ is negligible on the $1/\psi$ scale, which proves \eqref{eq:sharp-equilibrium-gap}.
\end{proof}

\subsection{A finite number of competing arbitrageurs}

The proportional allocation in the large-number competitive model can itself be
obtained as the limit of a strategic game with finitely many arbitrageurs. Let
$N\ge2$ arbitrage desks choose efforts $e_i\ge0$, with aggregate
$E:=\sum_{i=1}^N e_i$; write $e_{-i}$ for the vector of efforts of desks other than $i$. Let
\[
R(E):=C^\star(G(E)).
\]
For \(E>0\), desk \(i\) receives proportional share \(e_i/E\) of the
endogenous post-certification value pool and pays \(\psi e_i\):
\[
u_i(e_i,e_{-i})=\frac{e_i}{E}R(E)-\psi e_i,\qquad u_i(0,\mathbf0):=0.
\]
At a symmetric interior equilibrium, write $E_N>0$ for aggregate equilibrium effort. Differentiability of $R$ at $E_N$ gives
\[
\psi=
\left(1-\frac1N\right)\frac{R(E_N)}{E_N}
+\frac1N R'(E_N).
\]
The second term is a strategic prize-destruction effect. If $R'(E_N)<0$, a finite arbitrageur internalizes the reduction in the total opportunity caused by additional effort. For the following comparison, assume symmetric interior equilibria exist for every $N\ge2$, that $R$ is continuously differentiable and strictly decreasing in its positive-rent region $(0,E_*)$, with $R(0)>0$ and $R(E_*)=0$, and that
\[
F_N(E):=\left(1-\frac1N\right)\frac{R(E)}E+\frac1N R'(E)
\]
is strictly decreasing there. Then the corresponding roots $E_N$ increase with $N$, since $F_{N+1}(E)>F_N(E)$, and converge to the unique competitive root $R(E)/E=\psi$ for $\psi>0$. These are conditional comparative statics; the first-order equation alone does not establish existence of a Nash equilibrium.

A closed-form benchmark is obtained by using the globally defined clipped linear information impact
\[
G(E)=(\mathcal I_0-{\kappa_{{\rm crowd}}} E)_+,
\qquad {\kappa_{{\rm crowd}}}>0,
\]
with the constructive post-certification value
\[
C_0(\mathcal I)=a(\mathcal I-\mathcal I_*)_+,
\qquad a={v_{\rm KL}}(1-\beta).
\]

\begin{proposition}[Exact finite-$N$ linear benchmark]\label{prop:finite-linear}
Let $N\ge2$, $\mathcal I_0>\mathcal I_*$, $a,{\kappa_{{\rm crowd}}}>0$, and $\psi\ge0$. Writing
$\Delta\mathcal I=\mathcal I_0-\mathcal I_*$, define the rent pool for all $E\ge0$ by
$R(E)=a(\Delta\mathcal I-{\kappa_{{\rm crowd}}} E)_+$.
The proportional-share game has a unique symmetric equilibrium with positive gross rent. If $\psi>0$, it is also the unique symmetric equilibrium with positive aggregate effort. Its aggregate effort is
\[
E_N=\frac{a\Delta\mathcal I\,(N-1)}{N(a{\kappa_{{\rm crowd}}}+\psi)}.
\]
Define the equilibrium residual information by $\mathcal I_N^{\rm eq}:=\mathcal I_0-{\kappa_{{\rm crowd}}} E_N$.
The equilibrium remains in the positive-rent region, and its residual information satisfies
\[
\mathcal I_N^{\rm eq}-\mathcal I_*
=\Delta\mathcal I\frac{a{\kappa_{{\rm crowd}}}+N\psi}{N(a{\kappa_{{\rm crowd}}}+\psi)}.
\]
Equivalently,
\[
\frac{\mathcal I_N^{\rm eq}-\mathcal I_*}{\mathcal I_0-\mathcal I_*}
=
\frac{\psi}{a{\kappa_{{\rm crowd}}}+\psi}
+
\frac{a{\kappa_{{\rm crowd}}}}{N(a{\kappa_{{\rm crowd}}}+\psi)}.
\]
In particular, if $\psi=0$, then
\[
\mathcal I_N^{\rm eq}=\mathcal I_*+\frac{\mathcal I_0-\mathcal I_*}{N}.
\]
\end{proposition}

\begin{proof} See Appendix~\ref{app:finite-player-proof}.\end{proof}

The first term in the residual-information fraction comes from research cost; the second comes from having only finitely many competitors. Along the symmetric positive-rent equilibrium branch, a finite set of arbitrageurs leaves positive excess information even when research itself is costless; only the large-number competitive limit of this branch reaches the CAC frontier. In the symmetric linear benchmark, $1/N$ is also the Herfindahl index of effort shares, so the zero-cost residual-information fraction equals market concentration. With nonlinear crowding or asymmetric payoffs the relationship changes with the market technology.

\begin{figure}[t]
\centering
\includegraphics[width=0.94\linewidth]{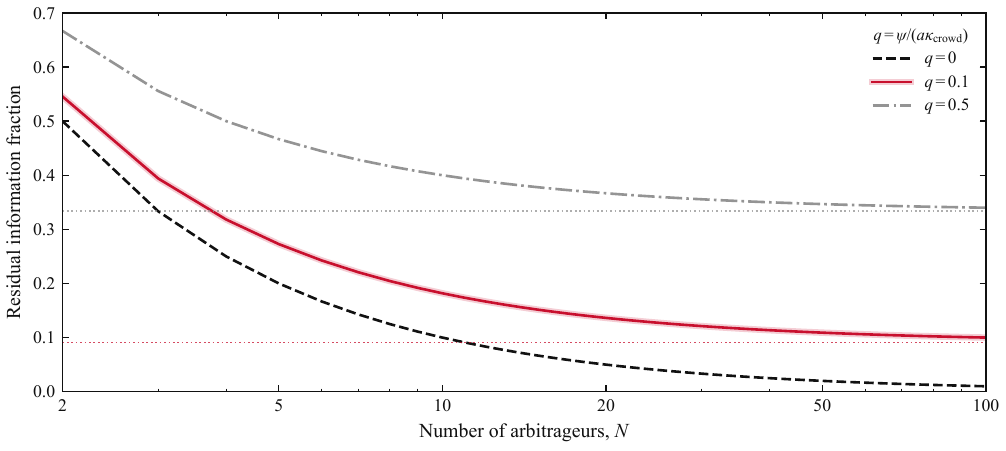}
\caption{Residual excess lifetime information along the symmetric positive-rent equilibrium branch of the exact linear-information strategic benchmark. With zero research cost the residual fraction is exactly $1/N$; positive research cost creates an additional limiting gap even as the number of competitors becomes large. Dotted lines mark $q/(1+q)$, where $q=\psi/(a\kappa_{\rm crowd})$; the pale-red outline highlights $q=0.1$.}
\label{fig:equilibrium-competition}
\end{figure}

In the zero-hurdle exponential Gaussian model,
$\mathcal I=\mathsf S^2h_{1/2}/(4\log2)$. Consequently, for opportunities that are actively
compressed by many low-cost arbitrageurs,
\[
\mathsf S^2h_{1/2}\longrightarrow 4\log2\,\mathcal I_*.
\]
This limit describes opportunities whose lifetime is actively compressed by many low-cost arbitrageurs.

\begin{remark}\label{heterogeneous-arbitrage-technology}
The equilibrium comparisons condition on implementation technology. Heterogeneity in research cost, crowding and competitive breadth can change their pooled sign; the opportunity's information budget alone is not a cross-sectional prediction of its rate of decay \cite{muravyev2025}.
\end{remark}

\section{Funding persistence as a lifetime signal}\label{retrospective-funding-study}

Funding persistence provides a descriptive illustration \cite{inan2025,he2026funding,lau2026funding,georgiev2026}. The benchmark uses BTC-USDT and ETH-USDT on Binance, Bybit, Gate, HTX, KuCoin and MEXC at eight-hour settlements. Events require 180 pre-event days, absolute funding above its trailing 95th percentile, a seven-day cooldown and complete seven-day follow-up; 543 of 545 events have finite scores. The $W=540$ pre-event settlements fit the no-intercept AR(1) $f_{j+1}=\phi f_j+\varepsilon_{j+1}$ by least squares, with residual sample standard deviation $\widehat\sigma_\varepsilon$. For $0<\widehat\phi<1$, set $\widehat\lambda_{\rm dec}=-\log\widehat\phi$ and $L_{\rm score}=f_t^2/[4\widehat\sigma_\varepsilon^2\widehat\lambda_{\rm dec}\mathcal I_{\rm crit}(0.05,0.10)]$. The threshold is $4.2819236753\ldots$ nats; normalization leaves ranks unchanged. Comparators are $z_{{\rm edge},t}=|f_t|/s_{f,t}$, where $s_{f,t}$ is the pre-event funding standard deviation, and the historical $t$ statistic of $\operatorname{sgn}(f_t)f_j$. The outcome is $Y_{7d,t}=\sum_{k=1}^{21}\operatorname{sgn}(f_t)f_{t+k}$.

Figure~\ref{fig:lifetime-benchmark} displays the estimates computed from the supplied event table: lifetime, edge and historical-$t$ associations are $0.4741$, $0.3833$ and $0.3606$. Residualizing lifetime midranks on edge midranks with an intercept leaves association $0.2777$ (95\% interval $[0.1228,0.4061]$). Intervals use 5,000 calendar-month cluster bootstrap resamples across 38 months, with seed 1729. On the same 543-event panel, the rank correlation between $L_{\rm score}$ and the sign-oriented AR(1) cumulative forecast $\widehat m_t=|f_t|\sum_{k=1}^{21}\widehat\phi^k$ is $0.9222$, rising to $0.9952$ when this forecast is divided by its conditional AR(1) standard deviation. The reproduction guide supplies the forecast-variance formula, event table and source identities. The score thus summarizes the AR(1) persistence channel. These retrospective quantities concern signed funding; net trading returns additionally require execution, hedge, financing and collateral costs.

\begin{figure}[t]
\centering
\includegraphics[width=0.96\linewidth]{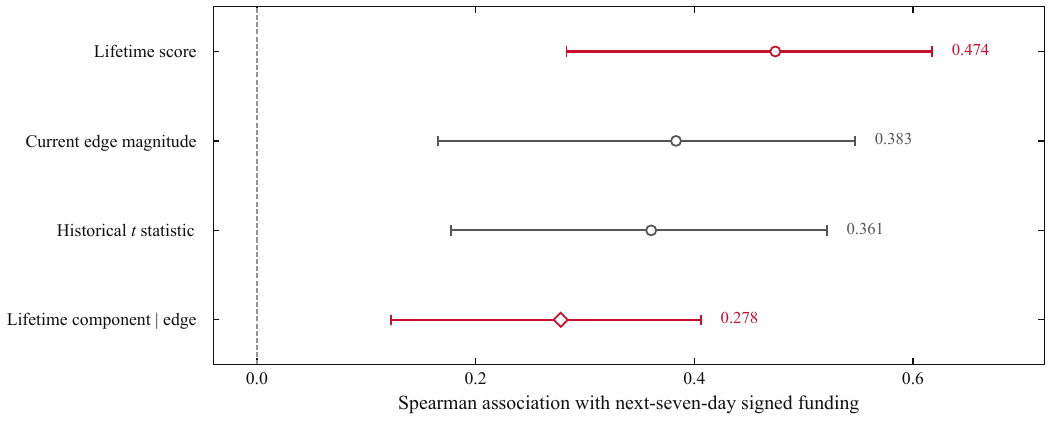}
\caption{Descriptive funding benchmark, 543 events. Points report Spearman associations with next-seven-day sign-oriented funding; horizontal whiskers show calendar-month cluster bootstrap 95\% intervals from 5,000 resamples. The last row residualizes lifetime-score rank on current-edge rank.}
\label{fig:lifetime-benchmark}
\end{figure}

\section{Discussion}\label{discussion}

CAC distinguishes evidence already accumulated, evidence that can still
arrive before economic death, and value remaining after certification.
A research team can use the frontier to screen a proposed validation
budget: a high current edge does not compensate for an information
lifetime too short to satisfy the declared reliability contract.
The screen concerns a new experiment; an experiment already in progress
also requires its current likelihood state.

For a portfolio manager, residual information and implementable cash
value are separate objects. Near the frontier, their logarithmic rates
can differ, and equal certifiability ratios need not imply equal economic
rankings. Trading limits, fees, implementation delays and the shape of
the remaining drift must enter the reward map. The theorems concern hypotheses within a specified drift experiment; realized portfolio profits lie outside their scope.

The theoretical conclusions also have explicit model limits. The
canonical likelihood must describe the joint observed experiment;
predictability of a signal alone does not make its likelihood contribution
disappear. Reward transfer requires uniform local behavior and control
away from the boundary. Ordered uncertainty preserves a least-profile
frontier only with a compatible reward order; nonordered alternatives
require a common minimax test. Estimated-model guarantees require uniform approximation bounds that include the null.

The funding exercise connects the lifetime perspective to an observable market quantity. Its retrospective associations characterize signed-funding persistence. A prospective trading evaluation would additionally specify execution, financing, hedging and portfolio constraints, yielding a distinct estimand from the persistence quantity studied here.
\section{Conclusion}\label{conclusion}

This perspective suggests treating validation speed as part of strategy design. A useful research rule must account for both the evidence needed to act and the opportunity remaining when that evidence arrives. Extending the analysis to richer observation models, implementable portfolio constraints and endogenous information acquisition would connect the certification frontier more closely to the allocation of arbitrage capital.

\section{Code and data availability}\label{code-and-data-availability}

The AlphaValue research package contains the numerical routines, synthetic examples, tests, editable figures and reproduction code for the theoretical and descriptive funding results reported here. The software archive is versioned at Zenodo under the concept DOI
\href{https://doi.org/10.5281/zenodo.23070742}{10.5281/zenodo.23070742}. Third-party funding histories remain subject to their source licenses; the reproduction guide identifies the required source revision and input paths. The derived event table, source-file identities and numerical references are included. The reproduction command recomputes the funding estimates and intervals offline; reconstructing the event table from raw histories requires the pinned inputs. The public bundle is scoped to this article. A clean editable installation uses \texttt{python -m pip install -e '.[test,plot]'} followed by \texttt{python reproduce/run.py}.

\fussy

\appendix

\section{Capacity with a type-I error constraint}\label{capacity-under-a-type-i-constraint}

With only a size constraint, let $c_\alpha(A)=\sup\Eone[(A-S)\delta]$ over $S\in\mathcal T_A$ and binary $\mathcal G_S$-measurable $\delta$ with $\Pzero(\delta=1)\le\alpha$. Independent randomization is available at time zero, so $c_\alpha(A)\ge\alpha A$. The economic value is $\mathsf{CAC}_\alpha(A)=\sigma^2c_\alpha(A)/(2\gamma)$.
\begin{theorem}[Nonasymptotic converse without a fixed power target]
For $A>0$, $r_\alpha=c_\alpha(A)/A$ obeys $A(1-r_\alpha)/2\ge\operatorname{kl}(r_\alpha,\alpha)$, and is bounded above by the unique root $\bar r\in[\alpha,1)$ of $A(1-\bar r)=2\operatorname{kl}(\bar r,\alpha)$.
\end{theorem}
\begin{proof}
For a feasible rule of reward $R$, set $r=R/A$, $p=\Pone(\delta=1)$, $q=\Pzero(\delta=1)$ and $S=A$ on nondeployment. Then $p\ge r$, $q\le\alpha$, and $(A-R)/2=\KL(\Pone^S\Vert\Pzero^S)\ge\operatorname{kl}(p,q)\ge\operatorname{kl}(r,\alpha)$ when $r\ge\alpha$. Apply this to a maximizing sequence and use continuity; the case $r_\alpha=\alpha$ is immediate. The root is unique because $A(1-r)-2\operatorname{kl}(r,\alpha)$ strictly decreases from positive to negative on $[\alpha,1)$.
\end{proof}
As $A\downarrow0$, $r_\alpha\to\alpha$: the limiting value comes from the permitted lottery alone. An evidence-added report can subtract $\alpha A$.

\paragraph{Constructive threshold.}\label{constructive-likelihood-threshold-lower-bound}
For $h>0$, put $\tau_h=\inf\{s:Z_s\ge h\}$. Brownian drift $\mu$ has first-passage cdf
\[
F_\mu(t;h)=\Phi((\mu t-h)/\sqrt t)+e^{2\mu h}\Phi((-\mu t-h)/\sqrt t).
\]
The unique $h_A$ with $F_{-1/2}(A;h_A)=\alpha$ gives the exact-level value $L_\alpha(A)=\int_0^A F_{1/2}(t;h_A)dt\le c_\alpha(A)$. The anytime boundary $h_\infty=\log(1/\alpha)$ has $\Pzero(\tau_{h_\infty}<\infty)=\alpha$.
\begin{theorem}[Asymptotic information cost of certification]
For fixed $\alpha\in(0,1)$, $A-c_\alpha(A)\to2\log(1/\alpha)$ as $A\to\infty$. Hence
\[
\frac{\mathsf{CAC}_\alpha(A)}{V^{PI}}
=1-\frac{\log(1/\alpha)}{\mathcal I_{\rm life}}+o(\mathcal I_{\rm life}^{-1}),
\qquad V^{PI}-\mathsf{CAC}_\alpha(A)\to\frac{\sigma^2}{\gamma}\log(1/\alpha).
\]
\end{theorem}
\begin{proof}
Under $\Pone$, the hitting time of $h=\log(1/\alpha)$ has mean $2h$, so $g_A:=A-c_\alpha(A)\le\Eone(\tau_h\wedge A)\to2h$. The converse gives $g_A/2\ge\operatorname{kl}(1-g_A/A,\alpha)$; boundedness of $g_A$ makes its right side converge to $h$, proving the matching lower bound.
\end{proof}

\section{Constrained stopping representation}\label{constrained-stopping-representation}

Setting $S=A$ on nondeployment makes the objective $A-\Eone S$, linking it to optimal group-sequential design \cite{ealesjennison1992,barberjennison2002,hamjasinzhao2026}. In the convexified class, a rule uses an $\mathcal G_S$-measurable deployment probability $\rho_{\rm dep}\in[0,1]$. Write $\mathcal T_{s,A}$ for stopping times in $[s,A]$ and $\mathbb E^{(1)}_{s,z}$ for the likelihood diffusion started at $Z_s=z$.
\begin{theorem}[Lagrange representation for the type-I constrained problem]
In the randomized stopping class,
\[
c_\alpha(A)=\inf_{\lambda_{\rm Lag}\ge0}\{\lambda_{\rm Lag}\alpha+V_{\lambda_{\rm Lag}}(0,0)\},\qquad
V_{\lambda_{\rm Lag}}(s,z)=\sup_{S\in\mathcal T_{s,A}}\mathbb E^{(1)}_{s,z}(A-S-\lambda_{\rm Lag}e^{-Z_S})_+.
\]
\end{theorem}
\begin{proof}
The stopped likelihood identity gives $\Ezero\rho_{\rm dep}=\Eone[e^{-Z_S}\rho_{\rm dep}]$. Maximizing the Lagrangian pointwise over $\rho_{\rm dep}\in[0,1]$ gives the positive-part payoff and weak duality. Randomized stopping times on a bounded horizon are compact in the Baxter--Chacon topology \cite{baxterchacon1977}; bounded rewards and uniformly integrable stopped likelihood ratios give continuity of reward and size in this Brownian model. The attainable set is therefore compact and convex. The no-deploy rule strictly satisfies size when $\alpha>0$, so separation gives a nonnegative multiplier and no duality gap.
\end{proof}
This construction uses the classical constrained-stopping framework \cite{nachman1980,kennedy1982}, with related deadline formulations in \cite{frazier2007deadline,dayanik2013deadline,ekstromwang2024,campbell2026deadline}. With generator $\mathcal L=\tfrac12\partial_{zz}+\tfrac12\partial_z$, the continuation equation is $\partial_sV+\mathcal LV=0$, with obstacle $(A-s-\lambda_{\rm Lag}e^{-z})_+$. In the positive-obstacle region, $(\partial_s+\mathcal L)(A-s-\lambda_{\rm Lag}e^{-z})=-1$, and the stopping boundary is an upper boundary $b_{\lambda_{\rm Lag}}(s)$. For $\lambda_{\rm Lag}>0$, $V_{\lambda_{\rm Lag}}(s,z)=V_1(s,z-\log\lambda_{\rm Lag})$ and $b_{\lambda_{\rm Lag}}=\log\lambda_{\rm Lag}+b_1$. Under smooth fit and the standard boundary regularity, the early-exercise representation and boundary equation are
\[
V_{\lambda_{\rm Lag}}(s,z)=\int_s^A\mathbb P^{(1)}_{s,z}\{Z_u\ge b_{\lambda_{\rm Lag}}(u)\}\,du,
\]
\[
A-s-\lambda_{\rm Lag}e^{-b_{\lambda_{\rm Lag}}(s)}
=\int_s^A\Phi\!\left(\frac{b_{\lambda_{\rm Lag}}(s)+(u-s)/2-b_{\lambda_{\rm Lag}}(u)}{\sqrt{u-s}}\right)du.
\]
For reliable CAC, add $\omega\ge0$: the stopping payoff becomes $(A-S+\omega-\lambda_{\rm Lag}e^{-Z_S})_+$, defining $V_{\lambda_{\rm Lag},\omega}$, and the dual objective is $\lambda_{\rm Lag}\alpha-\omega(1-\beta)+V_{\lambda_{\rm Lag},\omega}(0,0)$ under strong-duality hypotheses. Feasibility must first be checked by Theorem~\ref{thm:exact-gaussian-threshold}. For finite multipliers and $s<A$, waiting until any deterministic $u\in(s,A)$ has positive expected payoff by Gaussian full support. Immediate nondeployment cannot be optimal in this cost-free dual; an abandonment region requires observation cost or another binding constraint. The continuation state $(s,z)$ includes both time and log likelihood.

\subsection{Finite-player equilibrium proof}\label{app:finite-player-proof}

Proof of Proposition~\ref{prop:finite-linear}.
\begin{proof} In the positive-rent region the payoff is
\[
u_i=\frac{e_i}{E}a\Delta\mathcal I-(a{\kappa_{{\rm crowd}}}+\psi)e_i.
\]
Set $\mathcal V_0:=a\Delta\mathcal I$ and $k:=a{\kappa_{{\rm crowd}}}+\psi$. Define the aggregate rival effort by $E_{-i}:=\sum_{j\ne i}e_j$. Against $E_{-i}>0$, the first and second derivatives are
\[
\frac{\partial u_i}{\partial e_i}
=\frac{\mathcal V_0 E_{-i}}{(e_i+E_{-i})^2}-k,
\qquad
\frac{\partial^2 u_i}{\partial e_i^2}
=-\frac{2\mathcal V_0 E_{-i}}{(e_i+E_{-i})^3}<0.
\]
At a symmetric equilibrium \(e_i=e\) and \(E=Ne\), the first-order condition
gives \(e=\mathcal V_0(N-1)/(N^2k)\) and hence the stated \(E_N\). Moreover,
\[
\frac{{\kappa_{{\rm crowd}}} E_N}{\Delta\mathcal I}
=\frac{a{\kappa_{{\rm crowd}}}(N-1)}{N(a{\kappa_{{\rm crowd}}}+\psi)}<1,
\]
so the solution lies strictly above the CAC frontier. Substituting
$\mathcal I_N^{\rm eq}=\mathcal I_0-{\kappa_{{\rm crowd}}} E_N$ yields the residual-information identities. A unilateral deviation that pushes aggregate effort beyond the positive-rent region obtains zero gross rent and nonpositive net payoff, whereas the stated interior equilibrium gives a strictly positive payoff. Strict concavity in each player's own effort within the positive-rent region gives uniqueness among symmetric equilibria with positive gross rent. When $\psi>0$, any zero-rent profile with positive symmetric effort yields negative payoff and is improved by choosing zero. When $\psi=0$, additional zero-rent equilibria can exist, and no global uniqueness claim is made. \end{proof}

\section{Unknown scale and predictable regressors}\label{composite-uncertainty-and-unknown-scale}

\subsection{A pairwise KL necessary condition}\label{composite-pairwise-kl-condition}
Let $\mathfrak P_0,\mathfrak P_1$ be declared null and alternative laws on the observed sample space, and require $\sup_{Q_0\in\mathfrak P_0}Q_0(\delta=1)\le\alpha$ and $\inf_{Q_1\in\mathfrak P_1}Q_1(\delta=1)\ge1-\beta>\alpha$.
\begin{theorem}[Pairwise KL necessary condition]\label{thm:composite-kl-obstruction}
Every compatible pair $(Q_1,Q_0)\in\mathfrak P_1\times\mathfrak P_0$ with $Q_1\ll Q_0$ satisfies $\KL(Q_1\Vert Q_0)\ge\operatorname{kl}(1-\beta,\alpha)$.
\end{theorem}
\begin{proof} Data processing to $\delta$ gives the binary KL; on $p\ge1-\beta>\alpha\ge q$, it is minimized at $(1-\beta,\alpha)$.\end{proof}

\paragraph{AR(1) benchmark.}\label{g2-stationary-information-rate}
For $X_{t+1}=\phi X_t+\xi_{t+1}$, $Y_{t+1}=\theta X_t+\varepsilon_{t+1}$, assume independent centered Gaussian innovation sequences of variances $v_X,v_R>0$, $|\phi|<1$, and stationary $X$. The boundary companion has loading $b$ and the same nuisance parameters and initial law. The expected per-observation KL is
\[
\kappa_{\rm KL}(\theta,\phi,v_R,v_X;b)=\frac{(\theta-b)^2v_X}{2v_R(1-\phi^2)}.
\]
For a declared compact rectangle, let $\underline\kappa_{\rm KL}$ be its infimum. The minimum uses the smallest $|\theta-b|$ and $|\phi|$, smallest $v_X$ and largest $v_R$; if the persistence interval contains zero, its minimum is at zero. Every minimizing boundary companion must belong to the null family. If $\underline\kappa_{\rm KL}>0$, then $L_{\rm hor}\ge\lceil\operatorname{kl}(1-\beta,\alpha)/\underline\kappa_{\rm KL}\rceil$ is necessary. Marginal envelopes without a compatible joint pair do not prove impossibility.

\subsection{Sequential certification with predictable regressors}\label{predictable-design-anytime-achievability}
For $\mathcal F_t$-measurable $X_t$ before observing $Y_{t+1}$, assume $Y_{t+1}=\theta X_t+\varepsilon_{t+1}$, $\varepsilon_{t+1}\mid\mathcal F_t\sim N(0,v)$ with $v\le\bar v$. Test $\theta\le b$ against $\theta\ge b+d$, $d>0$. For $q>0$, define
\[
\ell_\alpha=\log(1/\alpha),\quad \lambda_q=\sqrt{2\ell_\alpha/(\bar v q)},\quad
Q_t=\sum_{s<t}X_s^2,\quad U_t=\sum_{s<t}X_s(Y_{s+1}-bX_s),
\]
and $E_t=\exp(\lambda_qU_t-\lambda_q^2\bar vQ_t/2)$, $\tau_q=\inf\{t:Q_t\ge q\}$, $\tau_E=\inf\{t:E_t\ge1/\alpha\}$.
\begin{theorem}[Predictable-design e-process certificate]\label{thm:predictable-certificate}
Under every null, $E_t$ is a nonnegative supermartingale and $\mathbb P_{\theta,v}(\tau_E<\infty)\le\alpha$. Under every alternative, if $\mathcal I_q=d^2q/(2\bar v)>\ell_\alpha$,
\[
\mathbb P_{\theta,v}(\tau_E>\tau_q,\ \tau_q<\infty)
\le e^{-(\sqrt{\mathcal I_q}-\sqrt{\ell_\alpha})^2}.
\]
Consequently $\mathcal I_q\ge[\sqrt{\log(1/\alpha)}+\sqrt{\log(1/\beta)}]^2$ gives power $1-\beta$ by $\tau_q$ if $\tau_q<\infty$ almost surely throughout the alternative class. Without reachability, only the displayed joint-event bound is asserted.
\end{theorem}
\begin{proof}
The conditional Gaussian mgf is at most one under the null, proving validity. Under an alternative, put $N_t=\sum_{s<t}X_s\varepsilon_{s+1}$; then $\log E_t\ge(\lambda_qd-\lambda_q^2\bar v/2)Q_t+\lambda_qN_t$. Failure at finite $\tau_q$ implies $N_{\tau_q}\le-c_qQ_{\tau_q}$, where $c_q=d-\lambda_q\bar v/2-\ell_\alpha/(\lambda_qq)=d-\sqrt{2\ell_\alpha\bar v/q}>0$. Optional sampling and Markov's inequality for the nonnegative supermartingale $\exp(-\zeta N_t-\zeta^2\bar vQ_t/2)$ bound this event by $\exp[-q(\zeta c_q-\zeta^2\bar v/2)]$. Optimizing at $\zeta=c_q/\bar v$ gives the result.
\end{proof}
Estimated $d$ or $\bar v$ require adding their confidence-set failure to the error budget. Unknown-variance inference is discussed in \cite{wangramdas2025,koningvanmeer2026}; the software includes fixed-design Studentization and Normal--Inverse-Gamma monitoring benchmarks, whose joint-model and clock assumptions must be checked separately.\label{nig-confidence-sequences}

\subsection{Discrete information--value identity}\label{discrete-g2-informationeconomic-conservation}
Use the stationary AR model above over $L_{\rm hor}$ periods. A position $q_t$ with $q_{-1}=0$ earns $\theta X_tq_t-\gamma q_t^2/2-\lambda_{\rm tr}(q_t-q_{t-1})^2/2$, with $\gamma>0$, $\lambda_{\rm tr}\ge0$, and terminal penalty $\eta_Tq_{L_{\rm hor}-1}^2/2$, $\eta_T\ge0$. Define backward from $K_{L_{\rm hor}}=\eta_T$, $\Lambda_{L_{\rm hor}}^q=0$,
\[
\begin{gathered}
D_t=\gamma+\lambda_{\rm tr}+K_{t+1},\quad a_t=\lambda_{\rm tr}/D_t,\quad
g_t(\phi)=(1+\phi\Lambda_{t+1}^q)/D_t,\\
K_t=\lambda_{\rm tr}(\gamma+K_{t+1})/D_t,\qquad \Lambda_t^q=\lambda_{\rm tr}g_t(\phi).
\end{gathered}
\]
The full-information and boundary feedbacks are $q_t^\theta=a_tq_{t-1}^\theta+\theta g_tX_t$ and $q_t^b=a_tq_{t-1}^b+bg_tX_t$.
\begin{theorem}[Discrete information--value identity]\label{thm:g2-conservation}
The expected regret of the boundary policy and return-channel information satisfy
\[
R_{L_{\rm hor}}^{\theta:b}
=\frac{v_X(\theta-b)^2}{2(1-\phi^2)}\sum_{t=0}^{L_{\rm hor}-1}D_tg_t^2
=\Pi_{L_{\rm hor}}\mathcal I_{L_{\rm hor},\rm return}^{\theta:b},
\]
where $\mathcal I_{L_{\rm hor},\rm return}^{\theta:b}=L_{\rm hor}v_X(\theta-b)^2/[2v_R(1-\phi^2)]$ and $\Pi_{L_{\rm hor}}=(v_R/L_{\rm hor})\sum_tD_tg_t^2$. With $\lambda_{\rm tr}=\eta_T=0$, $\Pi_{L_{\rm hor}}=v_R/\gamma$.
\end{theorem}
\begin{proof}
Backward completion of the Bellman square gives the feedbacks. Evaluated at the boundary policy's own inventory, its one-period Bellman loss is $D_t[(\theta-b)g_tX_t]^2/2$. Continuation values telescope, including terminal inventory cost. Take stationary expectations $\mathbb E X_t^2=v_X/(1-\phi^2)$ and divide by KL. When costs vanish, $D_tg_t^2=1/\gamma$.
\end{proof}
For contemporaneous innovation correlation $|\rho_{\rm corr}|<1$ with the same variances, joint Gaussian KL is $\mathcal I_{L_{\rm hor},\rm full}^{\theta:b}=\mathcal I_{L_{\rm hor},\rm return}^{\theta:b}/(1-\rho_{\rm corr}^2)$; the full-information price per nat is $(1-\rho_{\rm corr}^2)\Pi_{L_{\rm hor}}$. Zero correlation supplies the least-information member only when it belongs to the declared family.

\section{Random economic lifetime}\label{random-economic-death}

Let $H=A_T$, and consider a $\Ginfo$-stopping time $S\le H$ and a binary $\mathcal G_S$-measurable deployment decision $\delta$, with no deployment at death. Assume $\Eone H<\infty$, the stopped likelihood identity, and that observing the deadline adds no likelihood factor omitted from $Z$. Localization and $L^2$ optional sampling give $\KL(\Pone^S\Vert\Pzero^S)=\Eone S/2$. Data processing therefore implies the necessary obstruction
\[
\tfrac12\Eone H\ge\operatorname{kl}(1-\beta,\alpha)
\]
for size $\alpha$ and power $1-\beta>\alpha$. The mean is insufficient: lower tails control how often certification can finish. Let $\tau_{\rm cert}$ be the deployment time, with $\tau_{\rm cert}=+\infty$ on nondeployment. If $F(h)$ and $S_H(h)$ lower-bound $\Pone(\tau_{\rm cert}\le h)$ and $\Pone(H>h)$ for a fixed rule, the union bound gives $\Pone(\tau_{\rm cert}<H)\ge\sup_h[F(h)+S_H(h)-1]_+$ without independence. The canonical terminal power envelope gives the complementary upper bound $\Pone(\delta=1,S<H)\le\inf_h\{\pi_\alpha(h)+\Pone(H>h)\}\wedge1$ whenever the stopped experiment embeds in that canonical experiment.

For a common deadline law independent of the likelihood process under both hypotheses, let $H\sim\mathrm{Exp}(\varrho)$, $\varrho>0$, and deploy at $\tau_h=\inf\{s:Z_s\ge h\}$ only if $\tau_h<H$. Put $s_\varrho=\sqrt{1/4+2\varrho}$, $a_0=s_\varrho+1/2$ and $a_1=s_\varrho-1/2$. The Brownian first-passage Laplace transform gives
\[
\Pzero(\tau_h<H)=e^{-a_0h},\qquad
\Pone(\tau_h<H)=e^{-a_1h}.
\]
Thus $h=-\log\alpha/a_0$ gives exact size $\alpha$ and power $\pi^{\exp}_{\alpha,\varrho}=\alpha^{a_1/a_0}$. This constant-threshold rule meets the power target when that quantity is at least $1-\beta$. Setting $r=\log(1-\beta)/\log\alpha\in(0,1)$ and solving equality yields
\[
\varrho_{\rm bench}=\frac{r}{2(1-r)^2},\qquad
\frac{\Eone H}{2}=\frac{(1-r)^2}{r}.
\]
At $\alpha=0.05$, power $0.90$, the benchmark mean is about $26.47$ KL nats versus $4.28$ for a deterministic deadline. This constant-threshold rule gives a constructive benchmark; the optimal random-deadline frontier remains open.

\section{Sharp near-critical capacity}\label{app:sharp-capacity-proof}

This appendix proves Theorem~\ref{thm:near-critical-capacity} and the reward transfer used in Section~\ref{sharp-economic-capacity}. Throughout,
\[
A_*:=A_{\rm crit}(\alpha,\beta)=(z_{1-\alpha}+z_{1-\beta})^2,
\]
and
\[
c_A:=-\frac A2+\sqrt A\,z_{1-\alpha},
\qquad
\delta_A^{\rm NP}:=\mathbf1_{\{Z_A\ge c_A\}},
\qquad
\pi_\alpha(A):=\Phi(\sqrt A-z_{1-\alpha}).
\]
A randomized rule consists of a stopping time $S\le A$ and a deployment probability $0\le\delta\le1$ measurable with respect to $\mathcal G_S$, satisfying
\[
\Ezero\delta\le\alpha,\qquad \Eone\delta\ge1-\beta.
\]
Write
\[
u:=A-S,\qquad x:=Z_S-c_A.
\]

For $\rho_{\rm R}>0$, define the residual-moment capacity
\[
c^{[\rho_{\rm R}]}_{\alpha,\beta}(A)
:=\sup\Eone[u^{\rho_{\rm R}}\delta]
\]
over the same reliable rule class, and set
\[
\mathcal M_{\alpha,\beta}(\rho_{\rm R})
:=A_*^{\rho_{\rm R}}
\mathbb E\!\left[(N_{\rm G}+z_{1-\beta})_+^{2\rho_{\rm R}}\right],
\qquad N_{\rm G}\sim N(0,1).
\]

\begin{theorem}[Sharp residual-moment law]\label{thm:sharp-residual-moment}
For every $\rho_{\rm R}>0$,
\[
c^{[\rho_{\rm R}]}_{\alpha,\beta}(A_*)=0,
\]
and
\[
\lim_{\varepsilon\downarrow0}
[\log(1/\varepsilon)]^{\rho_{\rm R}}
 c^{[\rho_{\rm R}]}_{\alpha,\beta}(A_*+\varepsilon)
=2^{-\rho_{\rm R}}\mathcal M_{\alpha,\beta}(\rho_{\rm R}).
\]
The upper bound holds over randomized rules, and a binary conditional-error crossing rule attains the limit.
\end{theorem}

Theorem~\ref{thm:near-critical-capacity} is the case $\rho_{\rm R}=1$.

\subsection{Conditional Gaussian penalty}

For $N_{\rm G}\sim N(0,1)$ define
\[
\mathfrak q(u,x)
:=\mathbb E\!\left[(e^{-x-u/2-\sqrt u\,N_{\rm G}}-1)_+\right].
\]
Fix $A_{\max}>A_*$ and $\eta\in(0,1)$. There exists $c_\eta>0$ such that, for $0<u\le A_{\max}$,
\begin{equation}\label{eq:sharp-q-lower}
\mathfrak q(u,x)
\ge c_\eta u
\exp\!\left[-\frac{(1+\eta)x_+^2}{2u}\right].
\end{equation}
To obtain the bound, restrict $N_{\rm G}$ to
\[
\left[-x_+/\sqrt u-\sqrt u/2-2,
      -x_+/\sqrt u-\sqrt u/2-1\right].
\]
On this interval the payoff is at least $e^{\sqrt u}-1\ge\sqrt u$. A lower bound on the Gaussian density, together with
$(a+b)^2\le(1+\eta)a^2+(1+\eta^{-1})b^2$, gives \eqref{eq:sharp-q-lower}.

For a feasible rule at $A=A_*+\varepsilon$, define the Neyman--Pearson deficit
\begin{align}
\mathfrak d_A
&:=\Ezero[(e^{Z_A}-e^{c_A})(\delta_A^{\rm NP}-\delta)]\nonumber\\
&=\pi_\alpha(A)-\Eone\delta-e^{c_A}(\alpha-\Ezero\delta).
\label{eq:sharp-deficit}
\end{align}
The integrand is nonnegative pointwise, and hence
\[
0\le \mathfrak d_A
\le \pi_\alpha(A)-(1-\beta)
=O(\varepsilon).
\]
Conditioning at $S$ and retaining the paths that deploy below the terminal Neyman--Pearson cutoff yields
\begin{equation}\label{eq:sharp-deficit-conditional}
\mathfrak d_A
\ge\Eone[\delta\mathbf1_{\{u>0\}}\mathfrak q(u,x)].
\end{equation}
Under $\Pone$, conditional on $\mathcal G_S$, the increment $Z_A-Z_S$ is $N(u/2,u)$.

\subsection{Universal upper bound}

Let
\[
\Lambda_\varepsilon:=\log(1/\varepsilon)
\]
and define
\[
\mathcal B_\eta
:=\left\{u>0:\ u\Lambda_\varepsilon>
\frac{1+\eta}{2(1-\eta)}x_+^2\right\}.
\]
On $\mathcal B_\eta$, \eqref{eq:sharp-q-lower} and \eqref{eq:sharp-deficit-conditional} imply
\[
\Eone[\delta u\mathbf1_{\mathcal B_\eta}]=O(\varepsilon^\eta).
\]
If $0<\rho_{\rm R}\le1$, concavity on the measure with density $\delta\mathbf1_{\mathcal B_\eta}$ gives
\[
\Eone[\delta u^{\rho_{\rm R}}\mathbf1_{\mathcal B_\eta}]
\le
\{\Eone[\delta u\mathbf1_{\mathcal B_\eta}]\}^{\rho_{\rm R}}.
\]
If $\rho_{\rm R}\ge1$, use $u^{\rho_{\rm R}}\le A_{\max}^{\rho_{\rm R}-1}u$. Hence the contribution of $\mathcal B_\eta$, multiplied by $\Lambda_\varepsilon^{\rho_{\rm R}}$, vanishes uniformly over feasible rules.

On $\mathcal B_\eta^c$,
\[
\Lambda_\varepsilon^{\rho_{\rm R}}u^{\rho_{\rm R}}
\le
\left(\frac{1+\eta}{2(1-\eta)}\right)^{\rho_{\rm R}}
 x_+^{2\rho_{\rm R}}.
\]
Taking $\rho_{\rm R}=1$ first gives $\Eone[\delta u]=O(\Lambda_\varepsilon^{-1})$. Since $\delta$ is $\mathcal G_S$-measurable,
\[
\Eone[\delta(Z_A-Z_S)^2]
=
\Eone[\delta(u+u^2/4)]
=O(\Lambda_\varepsilon^{-1}).
\]
Conditional Gaussian moment bounds then yield, for every fixed $k>0$,
\[
\Eone[\delta|Z_A-Z_S|^k]\longrightarrow0.
\]
Consequently,
\[
\Eone[\delta(Z_S-c_A)_+^{2\rho_{\rm R}}]
\le
\Eone[(Z_A-c_A)_+^{2\rho_{\rm R}}]+o(1).
\]
Letting $\eta\downarrow0$ gives
\[
\limsup_{\varepsilon\downarrow0}
\Lambda_\varepsilon^{\rho_{\rm R}}
 c^{[\rho_{\rm R}]}_{\alpha,\beta}(A_*+\varepsilon)
\le
2^{-\rho_{\rm R}}\mathcal M_{\alpha,\beta}(\rho_{\rm R}).
\]
At $A=A_*$, \eqref{eq:sharp-deficit} gives $\mathfrak d_A=0$. Since $\mathfrak q(u,x)>0$ for every $u>0$ and finite $x$, \eqref{eq:sharp-deficit-conditional} forces $\delta u=0$ almost surely, so
$c^{[\rho_{\rm R}]}_{\alpha,\beta}(A_*)=0$.

\subsection{Attaining conditional-error crossing rule}

For $A>A_*$ define the terminal threshold calibrated to power $1-\beta$,
\[
\widetilde c_A:=A/2-\sqrt A\,z_{1-\beta},
\]
and its null rejection probability
\[
\alpha_A^{\rm term}
:=\Pzero(Z_A\ge\widetilde c_A)
=1-\Phi(\sqrt A-z_{1-\beta})<\alpha,
\qquad
\bar\alpha_A:=\frac{\alpha_A^{\rm term}}{\alpha}.
\]
Under $\Pzero$ the conditional terminal-rejection probability
\[
C_t^{\rm ce}
:=\Pzero(Z_A\ge\widetilde c_A\mid\mathcal G_t)
=
\Phi\!\left(
\frac{Z_t-\widetilde c_A-(A-t)/2}{\sqrt{A-t}}
\right),
\qquad t<A,
\]
is a bounded continuous martingale. Let $S_A$ be its first crossing of $\bar\alpha_A$ before $A$, with $S_A=A$ if no crossing occurs, and set
\[
\delta_A^{\rm ce}:=\mathbf1_{\{S_A<A\}}.
\]
Every terminal-positive path crosses the boundary. Optional sampling gives
\[
\Pzero(\delta_A^{\rm ce}=1)=\alpha,
\qquad
\Pone(\delta_A^{\rm ce}=1)
\ge\Pone(Z_A\ge\widetilde c_A)=1-\beta.
\]
Thus the rule is feasible.

As $\varepsilon\downarrow0$,
\[
1-\bar\alpha_A
\sim
\frac{\varphi_N(z_{1-\alpha})}{2\alpha\sqrt{A_*}}\,\varepsilon,
\qquad
\bigl[\Phi^{-1}(\bar\alpha_A)\bigr]^2
\sim2\Lambda_\varepsilon.
\]
On a crossing, let
\[
u_A^{\rm rem}:=A-S_A.
\]
Then
\[
\Phi^{-1}(\bar\alpha_A)\sqrt{u_A^{\rm rem}}
=Z_{S_A}-\widetilde c_A-u_A^{\rm rem}/2.
\]
Couple the experiments under $\mathbb P_1$ using $Z_t=B_t+t/2$ on
$[0,A_{\max}]$, and put $u_A^{\rm rem}=0$ when no crossing occurs. For all
$A$ sufficiently close to $A_*$, write
$k_A=\Phi^{-1}(\bar\alpha_A)>0$. There is a random variable $M_{\rm dom}$
with every finite moment, independent of $A$, such that on a crossing
\[
0\le k_A\sqrt{u_A^{\rm rem}}
=Z_{S_A}-\widetilde c_A-u_A^{\rm rem}/2\le M_{\rm dom}.
\]
Thus $k_A^2u_A^{\rm rem}\le M_{\rm dom}^2$, and the residual time tends to
zero on crossing paths. If $Z_{A_*}>c_{A_*}$, terminal rejection, and hence
a crossing, occurs for all sufficiently small $A-A_*$. Brownian continuity
then gives convergence to $(Z_{A_*}-c_{A_*})^2$. If $Z_{A_*}<c_{A_*}$, an
infinite sequence of crossings would have $S_A\to A_*$ and a negative
limiting right-hand side in the crossing identity, which is impossible.
The equality event has probability zero. Consequently,
\[
k_A^2u_A^{\rm rem}\longrightarrow (Z_{A_*}-c_{A_*})_+^2
\]
almost surely and, by domination, in every finite $L^q$. Since
$[\Phi^{-1}(\bar\alpha_A)]^2/(2\Lambda_\varepsilon)\to1$,
\[
\Lambda_\varepsilon^{\rho_{\rm R}}
\Eone[(u_A^{\rm rem})^{\rho_{\rm R}}\delta_A^{\rm ce}]
\longrightarrow
2^{-\rho_{\rm R}}\mathcal M_{\alpha,\beta}(\rho_{\rm R}).
\]
This matches the upper bound and proves Theorem~\ref{thm:sharp-residual-moment}. At $\rho_{\rm R}=1$,
\[
\mathcal K_{\alpha,\beta}
=\frac12\mathcal M_{\alpha,\beta}(1)
=\frac{A_*}{2}
\left(
[1+z_{1-\beta}^2](1-\beta)
+z_{1-\beta}\varphi_N(z_{1-\beta})
\right).
\]

\subsection{General residual rewards}\label{app:general-residual-rewards}

Let $\mathcal R_A:[0,A]\to[0,\infty)$ be measurable and uniformly bounded
for $A$ near $A_*$, with $\mathcal R_A(0)=0$. Suppose
$\mathcal R_A(u)/u^{\rho_{\rm R}}\to a_0>0$ uniformly as $A\downarrow A_*$
and $u\downarrow0$. Fix $h>0$ and choose $m>\rho_{\rm R}$. The
residual-moment theorem gives, uniformly over feasible rules,
\[
\mathbb E_1[\delta\mathbf1_{\{u\ge h\}}]
\le h^{-m}\mathbb E_1[\delta u^m]
=O([\log(1/\varepsilon)]^{-m}).
\]
Uniform boundedness makes the reward contribution on $u\ge h$ negligible
after multiplication by $[\log(1/\varepsilon)]^{\rho_{\rm R}}$. For each
$\eta>0$, the uniform local expansion gives, for sufficiently small $h$
and $A-A_*$,
\[
(a_0-\eta)u^{\rho_{\rm R}}\le\mathcal R_A(u)
\le(a_0+\eta)u^{\rho_{\rm R}},\qquad 0<u<h.
\]
The upper residual-moment bound and the attaining conditional-error rule
give matching upper and lower reward bounds. The discarded $u\ge h$ part
of the $\rho_{\rm R}$-moment is also negligible, since it is bounded by
$h^{\rho_{\rm R}-m}\mathbb E_1[\delta u^m]$. Letting $\eta\downarrow0$ yields
\[
\sup_{(S,\delta)}\Eone[\mathcal R_A(A-S)\delta]
\sim\frac{a_0\mathcal M_{\alpha,\beta}(\rho_{\rm R})}
{[2\log(1/\varepsilon)]^{\rho_{\rm R}}}.
\]
The value at $u=0$ contributes nothing because $\mathcal R_A(0)=0$.
This proves \eqref{eq:general-residual-reward}.

\subsection{Monitoring grids and implementation lead}\label{app:monitoring-scope}

The sharp law allows the preterminal monitoring date to approach the horizon with $\varepsilon$. A fixed finite grid with a final preterminal gap bounded away from zero defines a different local asymptotic regime. If every counted certificate must precede economic death by a fixed information-time lead $\ell>0$, then the observation horizon available for certification is $A-\ell$ and the feasibility boundary shifts to $A_*+\ell$. Calendar-time lead and information-time lead coincide only under a unit-speed information clock.

\subsection{Equal lifetime information, different cash capacity}\label{app:cash-capacity-counterexample}

On $[0,1]$ define
\[
\Delta^{(1)}(t)=1-t,
\qquad
\Delta^{(2)}(t)=(1-t)^2(1+b_{\rm poly}t+a_{\rm poly}t^2),
\]
where
\[
a_{\rm poly}:=\frac{135-5\sqrt{345}}8,
\qquad
b_{\rm poly}:=2-\frac{2a_{\rm poly}}5.
\]
Direct integration gives
\[
\Delta^{(1)}(0)=\Delta^{(2)}(0)=1,
\qquad
\int_0^1\Delta^{(1)}(t)\,dt
=
\int_0^1\Delta^{(2)}(t)\,dt
=\frac12,
\]
and
\[
\int_0^1[\Delta^{(1)}(t)]^2dt
=
\int_0^1[\Delta^{(2)}(t)]^2dt
=\frac13.
\]
With the same constant volatility, the two opportunities have the same
current edge, physical horizon, integrated future cash value, lifetime
information and $CR$. To make the near-critical family explicit, take the
common volatility $\sigma_A^2=1/(3A)$ and let $A\downarrow A_*$. Both
experiments then have information horizon $A$, with positive limiting
volatility and fixed terminal profiles satisfying the uniform reward
hypotheses. Their local crossing orders are $p=1$ and $p=2$, so \eqref{eq:cash-capacity-sharp} gives
\[
\mathsf{CAC}^{\rm cash}_{1}(A_*+\varepsilon)
\asymp[\log(1/\varepsilon)]^{-2/3},
\]
\[
\mathsf{CAC}^{\rm cash}_{2}(A_*+\varepsilon)
\asymp[\log(1/\varepsilon)]^{-3/5}.
\]
Lifetime information determines feasibility; cash value also depends on the payoff map.

\section{Proofs for unknown decay and model transfer}\label{app:unknown-decay-proofs}

\subsection{Ordered least-profile family}

Let $\underline\Delta$ and $\underline A_t$ be as in Section~\ref{ordered-unknown-decay}.  The standardized matched statistic
\[
\bar Z_T^{\rm lf}
:=
\frac{\int_0^T\underline\Delta(u)\,dY_u}
{\sigma^2\sqrt{\underline A}}
\]
is $N(0,1)$ under the boundary law.  Under alternative $\vartheta$ its mean is
\[
\frac{\int_0^T\underline\Delta(u)\Delta_\vartheta(u)\,du}
{\sigma^2\sqrt{\underline A}}
\ge\sqrt{\underline A}.
\]
Therefore $\mathbf1_{\{\bar Z_T^{\rm lf}\ge z_{1-\alpha}\}}$ has size $\alpha$ and worst-case power
$\Phi(\sqrt{\underline A}-z_{1-\alpha})$, attained at the least member.  No level-$\alpha$ rule can have larger worst-case power because the least member belongs to the family.

For the sharp capacity statement, apply the conditional-error crossing policy from Appendix~\ref{app:sharp-capacity-proof} on the deterministic least-profile clock.  Coupling each alternative and the least member with the same Brownian path gives
\[
\underline Z_t^{\vartheta}-\underline Z_t^{\rm lf}
=
\frac1{\sigma^2}\int_0^t
\underline\Delta(u)[\Delta_\vartheta(u)-\underline\Delta(u)]\,du
\ge0.
\]
Hence every least-profile crossing occurs no later under the other alternatives. Residual least-profile information is then pathwise larger. The same is true of the coupled payoff for a common nonnegative nondecreasing residual reward, or for the memberwise-dominating payoff maps specified in the theorem. The least member supplies the matching upper bound, so the sharp constants coincide with the simple least-profile constants.

\subsection{Least-favorable mixtures}

For the finite Gaussian experiment of Section~\ref{nonordered-unknown-decay}, define
\[
\mathcal D_\alpha
:=
\{\delta\in L^\infty(\Pzero):0\le\delta\le1,\ \Ezero\delta\le\alpha\}.
\]
For $w\in\mathcal W_J$,
\[
\Ezero[L_w\delta]=\sum_{j=1}^J w_j\mathbb E_j\delta
\]
is continuous and bilinear in $(\delta,w)$.  The decision set is weak-star compact and convex, and $\mathcal W_J$ is compact and convex.  The minimax theorem yields
\[
\sup_{\delta\in\mathcal D_\alpha}\min_j\mathbb E_j\delta
=
\min_{w\in\mathcal W_J}\sup_{\delta\in\mathcal D_\alpha}\Ezero[L_w\delta].
\]
The inner problem is a Neyman--Pearson test against the mixture likelihood $L_w$.

If two distinct alternatives satisfy $\|m_1\|^2=\|m_2\|^2=A_*$, each has individual maximal power $1-\beta$.  A common rule attaining that power under both would be a most-powerful size-$\alpha$ test for each member.  Continuous Gaussian likelihoods make each most-powerful test unique up to null sets, whereas distinct equal-norm vectors define distinct rejection halfspaces.  Hence the individual frontiers cannot be attained simultaneously.

\subsection{Same-contract approximation}

Let the policy notation be as in Section~\ref{estimated-model-transfer}.  Tightened approximate feasibility implies
\[
\operatorname{size}(\widehat\pi)
\le
\widehat{\operatorname{size}}(\widehat\pi)+\eta_P
\le\alpha,
\]
\[
\operatorname{pow}_j(\widehat\pi)
\ge
\widehat{\operatorname{pow}}_j(\widehat\pi)-\eta_P
\ge1-\beta
\quad\text{for every }j.
\]
Define
\[
\Pi(a,p):=\{\pi\in\Pi:\operatorname{size}(\pi)\le a,
\ \inf_j\operatorname{pow}_j(\pi)\ge p\},
\qquad
\mathcal U^\star(a,p):=\sup_{\pi\in\Pi(a,p)}\mathcal U(\pi).
\]
Assume $2\eta_P<\min(\alpha,\beta)$ and that
$\Pi(\alpha-2\eta_P,1-\beta+2\eta_P)$ is nonempty.
Suppose $\widehat\pi$ is $\varepsilon_{\rm opt}$-optimal for the tightened approximate problem with objective $\widehat{\mathcal U}$. Uniform memberwise reward error implies
$|\mathcal U(\pi)-\widehat{\mathcal U}(\pi)|\le\eta_V$.
Every policy in the twice-tightened true class is feasible for the once-tightened approximate problem. Comparing $\widehat\pi$ to an arbitrary such policy, and taking the supremum, gives
\[
\mathcal U(\widehat\pi)
\ge\mathcal U^\star(\alpha-2\eta_P,1-\beta+2\eta_P)
-2\eta_V-\varepsilon_{\rm opt}.
\]
Consequently,
\begin{align*}
\mathcal U^\star(\alpha,1-\beta)-\mathcal U(\widehat\pi)
&\le 2\eta_V+\varepsilon_{\rm opt}\\
&\quad+\mathcal U^\star(\alpha,1-\beta)
-\mathcal U^\star(\alpha-2\eta_P,1-\beta+2\eta_P).
\end{align*}
The last difference is the loss due to tightening the reliability contract. This argument guarantees the fixed worst-case criterion; each member's separate optimum remains outside this guarantee.

A sufficient probability bound is obtained from total variation, with convention $\operatorname{TV}(P,Q):=\sup_E|P(E)-Q(E)|$, where the supremum is over measurable events. If
\[
\max_{j\in\{0,1,\ldots,J\}}\operatorname{TV}(P_j,\widehat P_j)\le\eta_P,
\]
then every bounded history-measurable deployment probability changes by at most $\eta_P$.  If the reward range has length at most $V_{\rm span}$ and the approximate payoff differs pointwise by at most $\kappa_V$, one may take
\[
\eta_V=V_{\rm span}\eta_P+\kappa_V.
\]
Pinsker's inequality gives
\[
\operatorname{TV}(P_j,\widehat P_j)
\le
\sqrt{\tfrac12\KL(P_j\Vert\widehat P_j)}.
\]
For finite-dimensional Gaussian history laws with common positive-definite covariance $\Sigma$ and means $m,\widehat m$,
\[
\operatorname{TV}\{N(m,\Sigma),N(\widehat m,\Sigma)\}
=
2\Phi\!\left(
\frac{\|\Sigma^{-1/2}(m-\widehat m)\|}{2}
\right)-1.
\]
These inequalities put model-estimation error and the certification margin on the same probability scale.

\section{Notation and scope of the models}\label{app:notation}

Decorations distinguish mathematical roles; local model assumptions remain part of each result.
\begin{center}\small
\begin{tabular}{@{}p{0.29\textwidth}p{0.65\textwidth}@{}}
\toprule Symbol & Meaning and scope\\\midrule
$\mathbb P_0,\mathbb P_1;\ U_{\rm rnd}$ & Null/alternative laws; independent auxiliary uniform randomizer available at time zero.\\
$\alpha,\beta$ & Size and nondeployment bounds; sharp asymptotics require $0<\alpha<1-\beta<1$.\\
$A_t;\ H=A_T;\ A$ & Likelihood quadratic-variation clock; random terminal clock; deterministic horizon.\\
$t_{\rm inv}(s);\ Z_s$ & Stopped inverse clock; log likelihood in information time.\\
$S,\delta;\ \rho_{\rm dep}$ & Stopping time and binary decision; convexified deployment probability in $[0,1]$.\\
$c_{\alpha,\beta}(A)$ & Reliable residual-information capacity, reported as zero below feasibility.\\
$A_*,\mathcal I_*$ & Critical information time and KL information, $A_*=2\mathcal I_*$.\\
$\Delta;\ c,p$ & Boundary drift gap; terminal coefficient (vector or scalar as stated) and crossing order $p>0$.\\
$\lambda_{\rm dec};\lambda_{\rm Lag},\omega$ & Decay rate; nonnegative size/power dual multipliers.\\
$\lambda_q;\lambda_{\rm tr}$ & Predictable-certificate tilt; quadratic turnover-cost coefficient.\\
$v_{\rm KL};\ v_R,v_X$ & Value per KL nat; return and predictor innovation variances.\\
$\rho_{\rm corr};\rho_{\rm R}$ & Innovation correlation; positive residual-reward exponent.\\
$\kappa_{\rm crowd};\kappa_{\rm KL}$ & Linear crowding sensitivity; stationary pairwise information rate.\\
$L_{\rm hor};\ L_{\rm score}$ & Discrete horizon; empirical lifetime score with descriptive scope.\\
$L_j,L_w$ & Alternative and mixture likelihood ratios in the finite robust experiment.\\
$\mathcal U(\pi);\eta_P,\eta_V$ & Worst-member transfer value; uniform probability/reward approximation bounds.\\
$\mathcal K_{\alpha,\beta};\mathcal M_{\alpha,\beta}(\rho_{\rm R})$ & Sharp inverse-logarithmic constant; critical Gaussian positive-part moment.\\
$\mathcal R_A(u);\ K_p^{\rm unit}$ & Residual reward with the stated uniform local expansion; unit-exposure reward coefficient.\\
\bottomrule
\end{tabular}
\end{center}

\end{document}